\documentclass[18pt]{amsart}
\usepackage{amsmath,amssymb,amsthm,hyperref}
\usepackage{stmaryrd}
\usepackage{graphicx}
\usepackage{tikz} 
 
 \newtheorem{theorem}{Theorem}[section]
 \newtheorem{corollary}[theorem]{Corollary}
 
 \newtheorem{lemma}[theorem]{Lemma}
 \newtheorem{proposition}[theorem]{Proposition}
 \newtheorem{example}[theorem]{Example}
 
 {\theoremstyle{definition}
 \newtheorem{definition}[theorem]{Definition}
 \newtheorem{remark}[theorem]{Remark}
 
 \newtheorem*{mmseconj}{MMSE Conjecture}
 \newtheorem*{compmonoconj}{Completely Monotone Conjecture}
 
 \newtheorem*{proof1}{Proof of Theorem \ref{T2}}
 \newtheorem*{proof2}{Proof of Theorem \ref{T1}}
 \newtheorem*{proof3}{Proof of Theorem \ref{T3}}}

 \newcommand{\MMSE}{\mathit{MMSE}}
 \newcommand{\mmse}{\mathit{mmse}}
 \newcommand{\sMD}{\textit{SMD}}
 \newcommand{\E}{\mathbb{E}}
 \renewcommand{\P}{\mathbb{P}}
 \newcommand{\R}{\mathbb{R}}
 \newcommand{\N}{\mathbb{N}}
 \newcommand{\Z}{\mathbb{Z}}
\newcommand{\dint}{{\rm d}}

\title{Derivatives of entropy and the MMSE conjecture}
\author{Paul Mansanarez}
\address{
Paul Mansanarez, Universit\'e de Rennes/Universit\'e libre de Bruxelles, France/Belgium. E-mail:
paul.mansanarez@ulb.be}
\author{Guillaume Poly}
\address{Guillaume Poly, Universit\'{e} de Rennes, France. E-mail:guillaume.poly@univ-rennes.fr}
\author{Yvik Swan}
\address{
Yvik Swan, Universit\'e libre de Bruxelles, Belgium. E-mail:
yvik.swan@ulb.be}

\begin{document}

\begin{abstract}
{We investigate the entropy $H(\mu,t)$ of a probability measure $\mu$ along the heat flow and more precisely we seek for closed algebraic representations of its derivatives. Provided that $\mu$ admits moments of any order, it is indeed proved in \cite{Ver} that $t\mapsto H(\mu,t)$ is smooth, and in \cite{Led} that its derivatives at zero can be expressed into multivariate polynomials evaluated in the moments (or cumulants) of $\mu$. In the seminal contribution \cite{Led}, these algebraic expressions are derived through $\Gamma$-calculus techniques which provide implicit recursive formulas for these polynomials. Our main contribution consists in a fine combinatorial analysis of these  inductive relations and for the first time  to derive closed formulas for the leading coefficients of these polynomials expressions.~\\
~\\
Building upon these explicit formulas we  revisit the so-called  ``$\MMSE$ conjecture'' from  \cite{Ver} which asserts that two distributions on the real line with the same entropy along the heat flow must coincide up to translation and symmetry. Our approach enables us to provide new conditions on the source distributions ensuring that the $\MMSE$ conjecture holds and to refine several criteria proved in \cite{Led}. As illustrating examples, our findings cover the cases of uniform and Rademacher distributions, for which previous results in the literature were inapplicable.}


\end{abstract}

\maketitle


\section{Introduction}

\subsection{Entropy and Fisher information along the heat flow}

 Let $X$ be an arbitrary real random variable, and $N$ an independent  standard Gaussian. The \emph{heat flow}  starting from $X$  is the process
 \begin{equation}\label{eq:heatfl}
     X_t = X + \sqrt{2t} N, \quad t>0;
 \end{equation}
    this denomination follows because $X_t$'s probability density function  $f_t(x)$     satisfies the heat equation $\partial _t f_t(x) = f_t''(x)$ for all $t>0$ and all $x\in \mathbb{R}$ ($'$ denotes  the spatial derivative). Closely related to the heat flow is the so-called  \emph{Gaussian channel} 
  \begin{equation}\label{eq:gaucha}
 Y_{s} = \sqrt{s} X + N, \quad  s >0
 \end{equation}
  (the  parameter $s$ in this context is called  the  \emph{signal-to-noise ratio}); the  density of $Y_{s}$ is obtained from a simple reparameterisation of the density $f_t(x)$ of $X_t$.  The heat flow  \eqref{eq:heatfl} (or equivalently the Gausian channel \eqref{eq:gaucha}) is the most basic model for  an observed  signal obtained by perturbing  a random unobserved source $X$ through an independent additive Gaussian noise and  the problem of studying the relation between   $X$ and   $X_t$ (or, equivalently, $Y_{s}$) is a  fundamental  problem in   information and  signal theory, see \cite[Chapter 9]{Cover2005} or \cite{Ver} for an overview.

Let $v_t = \log f_t$  so that $v_t' = f_t'/f_t$ is the model's score function along the heat flow. Entropy and Fisher information being canonical measures of a signal's disorder and information content, we introduce the functionals 
\begin{align*}
     H(X, t) = - \mathbb{E}[v_t(X_t)] := -\int_\R f_t(x) \log f_t(x) \, \dint x 
     \mbox{ and } I(X, t)  =  \mathbb{E}[
{(v'_t(X_t))^2}] 
\end{align*}
which express the entropy and Fisher information of $X_t$ as a function function of the variance of the added Gaussian noise. From the heat equation, these two functions are seen to be related through the famous \emph{De~Bruijn identity} 
\[\partial_t H(X, t) = I(X, t) \quad t>0\]
(first reported in  \cite[equation (2.12)]{stam1959some}) which indicates that the rate of increase of the entropy along the heat flow is precisely the Fisher information.    In particular entropy is increasing along the heat flow. 

Second and third order derivatives  of $H(X, t)$ were used in \cite{costa1985new, dembo1989simple} to prove  concavity in $t$ of  the entropy power function.  More recently, 
\cite{Ver}  obtain explicit formulas  for the derivatives up to order 5. In all these cases the derivatives of entropy (and thus of Fisher information) are  polynomials in some  conditional moments of the underlying signal $X$. To the best of our knowledge, the only complete description of the derivatives of $H(X, t)$ is provided by      \cite{Led} where  the $\Gamma$-calculus formalism introduced in   \cite{LedG} is used to characterise all higher order derivatives of $H(X, t)$ as expectations of certain implicit polynomials evaluated on moments of the score function of $X_t$.   More precisely, the following result holds.

\begin{theorem}[Theorem 1, \cite{Led}]\label{entr-deri}
    If $X$ admits moments at all orders then
    \begin{equation} \label{eq:partiLed}
            \partial_t^n H(X,t) = (-2)^{n-1}  \mathbb{E} \left[ V_n^2+{R}_n \big (V_2, \ldots, V_{n-1}\big ) \right]
    \end{equation}
    where $V_k := v_t^{(k)}(X_t)$  with $v_t^{(k)}$ the $k$th derivarive of $v_t$ (in space) and  $R_n \in \Z[T_2, \ldots, T_{n-1}]$ is a multivariate polynomial in $n-2$ variables defined in terms of iterated gradients, see Theorem \ref{thom1led}. 
\end{theorem}

Despite the intense scrutiny to which it was submitted over the past decades,  several fundamental questions remain open regarding  $H(X, t)$ and its derivatives. 
A first conjecture concerns the signs of the derivatives of the entropy. Clearly Fisher information is positive, hence so is (from De Bruijn's identity) the first derivative of entropy along the heat flow. Similarly the second derivative of $I(X, t)$  can be computed explicitly and is easily seen to be positive. This is no longer the case at higher orders, in fact already at at order $n=2$ direct computations from formula \eqref{eq:partiLed}  yield 
\begin{equation*}
         \partial^2_t I(X,t)  = \E \big [ V_3^2-2V_2^3 \big ]. 
 \end{equation*}
  Positivity of this expression  is not immediate at first sight. However, in \cite{cheng2015higher}, by expressing the $V_i$ as polynomials in ratios $F_j:=f_t^{(j)}(X_t)/f_t(X_t)$, and using integration by parts, the above is seen to be equivalent to 
\begin{equation*}
        \partial^2_t I(X,t)  = \E \left [ \left (F_3-F_2F_1+\frac{1}{3}F_1^3 \right )^2 +F_1^6 \right  ].
\end{equation*}
So far it is not yet  known whether 
\begin{equation}\label{deri-I}
(-1)^{k}\partial_t^k I(X, t) \ge 0
\end{equation}
for $k \ge 4$; this is the \emph{``completely monotone conjecture''}.

\begin{compmonoconj}
    Let $X$ be a random variable such that $t\longmapsto I(X,t)$ is infinitely differentiable. Then the function $t \longmapsto I(X,t)$ is completely monotonous, that is 
    \begin{equation*}
        (-1)^n\partial_t^n I(X,t) \geq 0
    \end{equation*}
    for every $t>0$ and non-negative integer $n$.
\end{compmonoconj}

A second question is related to minimal conditions under which $H(X, t)$ or, equivalently, $I(X, t)$, characterizes the law of the underlying signal  $X$. Such conditions are indeed not yet completely understood. This question was coined the \emph{``MMSE conjecture''} in  \cite{Ver} (see Section \ref{sec:thesetting} for an explanation of the denomination). 
 
\begin{mmseconj}[Entropic form]\label{MMSEConj-entr}
Let $X$ and $Y$ be two real random variables. If $H(X,\cdot)=H(Y,\cdot)$ then for some $a\in \R$, $\P_{Y}=\P_{X+a}$ or $\P_{Y} = \P_{-X+a}$. 
\end{mmseconj}

\noindent This question is  reminiscent of the deep and classical \emph{moment problem}  which consists in finding conditions ensuring the moment determinacy of probability measures (replacing the equality of moments by the equality of  entropies along the heat flow). 
See Section \ref{sn: a contd} for a more detailed discussion. \medskip

These are not the only interesting conjectures related to this topic, and we refer to \cite{ledoux2022differentials} for an overview along with up-to-date references on several other open questions in the same spirit.

\subsection{Our contribution regarding the derivatives of the entropy}


As described below Proposition 9 in \cite{Led}, understanding the combinatorial structure of the polynomials $R_n$ from  Theorem \ref{entr-deri} is of interest.  
Relatively straightforward, though tedious, computations lead to the following expressions for small values of $n$:
\begin{align*}
R_2 & = 0, \\
R_3 & = -2T_2^3,\\ 
R_4 & = - 12T_2T_3^2 + 6T_2^4,\\
R_5  &= - 20T_2T_4^2 - 30T_3^2T_4 + 120T_2^2T_3^2 -24T_2^5, \\
R_6  &= - 30T_2T_5^2 - 120T_3T_4T_5 + 900T_2T_3^2T_4 + 300T_2^2T_4^2 - 30T_4^3 - 1200T_2^3T_3^2 \\ & \qquad + 210T_3^4 + 120T_2^6
\end{align*}
(only $R_6$ is new in the above list). 

The main result of this paper has a combinatorial flavour and provides the leading terms in the polynomial expansions from \eqref{eq:partiLed}.

\begin{theorem}\label{Main4} 
Let $n$ be an integer greater than $2$ and let $R_n \in \mathbb{R}[T_2, \ldots, T_n]$ be the polynomial defined in  \eqref{eq:partiLed}. Then 
\begin{equation}\label{formRn} 
R_{n+1} = -n(n+1)T_2T_n^2 + a_nT_n + \sum_{k=1}^{n-2} a_{n,k}T_{n-k}
\end{equation}
where $a_{n,k}$ is a polynomial in $\mathbb{Z}[T_2,\ldots,T_{n-k}]$ with degree equal to $2+k$ and $a_n$ is a polynomial in $\mathbb{Z}[T_2,\ldots,T_{n-1}]$  defined  for all $2m > 7$  by 
\[a_{2m} = -(2m+1)  \begin{pmatrix} 2m \\ m \end{pmatrix} T_{m+1}^2 - (4m+2)\sum_{k=3}^{m}   \begin{pmatrix} 2m \\ k-1\end{pmatrix} T_{2m+2-k}T_{k}\] 
and 
\[a_{2m+1} = -2m  \begin{pmatrix} 2m \\ m \end{pmatrix} T_{m+2}T_{m+1}- 4(m+1)\sum_{k=3}^{m}   \begin{pmatrix} 2m+1 \\ k-1\end{pmatrix} T_{2m+3-k}T_{k}.\]
\end{theorem}

\begin{remark}
    In addition to the theoretical results, another contribution of this paper is to provide an algorithm allowing to compute efficiently the polynomials $R_n$ for moderate orders  (see \cite{AlgoG}). More details are given in Section \ref{Combi}. This allows for the first time to perform exploratory analysis on the combinatorial structures of the polynomials.  
    \end{remark}

\begin{remark}
One could write also
\[R_{n+1} = \sum_{k=0}^{n-2} a_{n,k}T_{n-k}\]
where $a_{n,0} = -n(n+1)T_2T_n + a_n$. Futhermore, in \cite{Led}, it has been shown that $a_{n,n-2} = (-1)^{n+1}n! \, T_2^n$.
\end{remark}

\begin{remark}
The remaining $a_n$ not explicited in Theorem \ref{Main4} are
\begin{align*}
a_7 & = -560T_4T_5 -336T_3T_6, \\ 
a_6 & = -140T_4^2 -210T_3T_6, \\
a_5 & = -120T_3T_4, \\
a_4 & = -30T_3^2, \\
a_3 & = 0,
\end{align*}
which all have degree $2$ except $a_3$.
\end{remark}

 \begin{example}\label{ExR9}
For the sake of illustration, the 42-term long expression for $R_9$ is provided in Appendix \ref{appen-pol}.   We can verify the statement of Theorem~\ref{Main4} on this example, because  
\begin{equation}
\begin{aligned}
R_9 = 
& - \underbrace{72}_{=8\times 9}T_2T_8^2 + \big (- \underbrace{630}_{=9 \begin{pmatrix} 8 \\ 4 \end{pmatrix}}T_5^2 - \underbrace{504}_{=18\begin{pmatrix} 8 \\ 2 \end{pmatrix}}T_3T_7- \underbrace{1008}_{=18\begin{pmatrix} 8 \\ 3 \end{pmatrix}}T_4T_6 \big )T_8 \\ \\
 & + a_{8,1}T_7 + a_{8,2}T_6+a_{8,3}T_5 + a_{8,4}T_4 + a_{8,5} T_3 + a_{8,6}T_2
\end{aligned}
\label{R9}
\end{equation}
where the $a_{8,k}$ are given in the appendix (see Example \ref{R9ank}).


\end{example}


We stress that the combinatorial analysis behind  Theorem \ref{Main4} is very involved;  there seems to be little hope in obtaining a general formulation for $R_n$ at all orders.

\subsection{Outline of the proof of Theorem \ref{Main4}}

In this subsection, we give an overview of the  proof strategy  behind Theorem \ref{Main4}. The proper proof is provided in Section \ref{Combi}. 
The main goal is to describe the polynomial $R_n$ as completely as possible. To that purpose, we want to compute its coefficients. This involves the following steps: 
\begin{itemize}
\item[$\bullet$] First, we retrieve from \cite{Led} an induction formula between $R_{n+1}$ and $R_n$ (see Lemma \ref{inductionR}, relation \eqref{RecFormula}).
\item[$\bullet$] Then, from \eqref{RecFormula}, we obtain linear recurrence relations between coefficients of $R_{n+1}$ and coefficients of $R_n$ (see Proposition \ref{relcoefRn}).
\item[$\bullet$] Finally, we solve those recurrence relations. 
\end{itemize}
However, due to the complexity of the combinatorics, the induction relations become rapidly intractable.  Fortunately, we manage to get some information on degrees of monomials in $R_{n+1}$, and are able to compute explicitly the coefficients of the leading terms  (i.e.\ the monomials that are non constant polynomials for the indeterminate $T_n$)  of $R_{n+1}$ (see Lemmas \ref{P1}, \ref{P2}, \ref{P4} and \ref{an}). 

The procedure to compute the coefficients highlights an algorithmical resolution. Since some further results and notations are necessary to properly delve into this algorithm, more details are provided in Section \ref{AlgoScheme}.

\subsection{Structure of the paper} 

The rest of the  paper is structured as follows. In the second section, we recall some useful facts about  the cumulants of real random variables, as well as conditions for the moment determinacy. Then, we state our main contribution regarding the MMSE conjecture, which proceeds from our main result Theorem \ref{entr-deri}. In the third section, we describe the form of the polynomials and prove Theorem \ref{Main4}. Finally an appendix contains proofs of results concerning cumulants as well as more technical and routine proofs of results concerning the polynomials $R_n$.

\section{The MMSE conjecture}
\label{sn: a contd}

In order to state our main results we first gather some well-known facts about moments and cumulants as well as some standard criteria ensuring the moment-determinacy of probability measures. Indeed, a recurrent assumption throughout the whole article will be to assume that the source distributions are symmetric and that their distribution is characterised by the moments, or equivalently the cumulants. In order to write the assumptions in a more concise way, we therefore introduce the following definition. 
\begin{definition}\label{smd}
    A random variable $X$ is \textbf{\sMD}~~if and only if:
\begin{itemize}
    \item $X\stackrel{\text{Law}}{=}-X$,

    \item $X$ admits moments of every order,

    \item $\forall p \in \mathbb{N}^*, \, \mathbb{E}[X^p]=\mathbb{E}[Y^p]$ implies that $X\stackrel{\text{Law}}{=}Y$.
\end{itemize}

\end{definition}

\begin{remark}
    [{When is a law moment-determined?}]\label{mdet}
Take $\mu$ and $\nu$ two probability measures on the real line. An old and classical problem in probability theory consists in finding conditions on $\mu$ ensuring that
$$\left(\forall p\in\mathbb{N},\,\int_{\mathbb{R}}x^p \mu( \dint x)=\int_{\mathbb{R}}x^p \nu(\dint x)\right)\Rightarrow \mu=\nu.$$
It is well known that without imposing condition on $\mu$, the above
implication is in general not true. For instance, if
$X\sim\mathcal{N}(0,1)$, then $\mu=\mathbb{P}_{X^3}$ is not
moment-determined. Another well-known example of distribution which is
not moment-determined is the log-normal distribution (see
\cite{Stoyanov2000}).  Criteria providing sufficient conditions for
moment determinacy are well studied, see e.g.\ \cite{Lin2017}.
\end{remark}

\medskip

\subsection{Background on moments and cumulants}\label{S2}

Cumulants are alternative semi-invariants for distributions, related to the moments (see Proposition \ref{Cumul1} below).
We gather here some facts concerning cumulants of real random variables. These are mostly well-known results; for sake of completeness, we provide a self-contained presentation of this setting.
\begin{definition}
Let $X$ be a real random variable admitting moments of all orders. For $n$ a positive integer, we denote $K_n(X)$ the quantity $(-i)^n(\ln\varphi_X)^{(n)}(0)$ where $\varphi_X$ is the characteristic function of $X$. We call $K_n(X)$ the cumulant of order $n$ of $X$. We can also define it as the $n$-th derivative of the logarithm of the moment generating function $t\mapsto \E[e^{tX}]$ of $X$, also called the cumulant generating function of $X$.
\end{definition}

The following three results are classical and, moreover, easy to prove directly. Brief proofs are provided in the appendix.

\begin{proposition}\label{Cumul2}
Let $X$, $Y$ be  independent real random variables admitting moments of all orders. Then for all $\lambda \in \mathbb{R}$ and $a\in \mathbb{R}$
\[K_n(X+\lambda Y+a) = K_n(X) + \lambda^n K_n(Y)\]
for all $n \in \mathbb{N}^\star$, $n\geq 2$.
\end{proposition}

Proposition \ref{Cumul2} follows from the independence of $X$ and
$\lambda Y$, using it to write the characteristic function of
$X+\lambda Y$ as the product of the characteristic functions of $X$
and of $\lambda Y$. An immediate corollary for symmetric random
variables, which are of interest in the rest of the article, is the
following.

\begin{corollary}\label{CumulSym}
Let $X$ be a real symmetric random variable admitting moments of all orders. Then 
$K_{2n+1}(X) = 0$ for all $n\in \mathbb{N}^*$.
\end{corollary}

As mentioned above, cumulants are related to the moments through the
next result (a proof of which can be found in \cite{Doring2022}).
\begin{proposition}\label{Cumul1}
Let $X$ be a real random variable admitting moments of all orders and let $n$ be a positive integer. Then
\[K_n(X) = \E[X^n] -   \sum_{m=1}^{n-1} \binom{n-1}{m-1} K_{m}(X)\E[X^{n-m}]\] 
and, in particular,  moments can be recovered from  cumulants and vice versa. 
\end{proposition}

Our next result is
well-known, but not that evident to prove using elementary means nor
to get an easy reference; we therefore provide a proof in the
appendix. 

\begin{proposition}\label{CumulID}
Let $X$ be a real random variable with an infinite divisible distribution and admitting moments of all orders. Then the cumulants of even orders of $X$ are non-negative. 
\end{proposition}

We conclude with the computation of cumulants in a number of specific examples; for the same reason as previously, proofs are provided in the appendix. 

\begin{proposition}\label{prop-cumul}
    Let $n$ be a positive integer. Then the cumulant of order $n$ of:
    \begin{enumerate}
        \item the \textbf{uniform} distribution on $[-1,1]$ is ${2^nB_n}/{n}$ ;
        \item the \textbf{Laplace} distribution of parameter $(0,b)$ is ${(2n)!}b^{2n}/{n}$ ;
        \item the \textbf{Rademacher} distribution is ${2^{n}(2^{n}-1)}B_{n}/{n}$,
    \end{enumerate}
    where $B_n$ denotes the $n$-th Bernoulli number.
\end{proposition}




\begin{remark}
 A well-known property of the
Bernoulli numbers is that  $B_{2n} =
(-1)^{n+1}|B_{2n}|$  for every integer $n$; in particular, this sequence alternates in sign    (for instance, see \cite{Ireland1990}).
\end{remark}


\begin{remark}[Remark \ref{mdet}, continued]
   From Proposition \ref{Cumul1} we know that moments and cumulants are related through specific polynomial relations. In fact, being moment-determined is equivalent to being characterised by one's cumulants. We refer to \cite{Led}  for  an example of two different SMD distributions whose cumulants have the same absolute values of all orders  (with different signs). 
Note that it is still an open problem to determine whether a given sequence of numbers corresponds to a sequence of cumulants, see \cite{Nardo2009}.  
\end{remark}

\subsection{The setting around the MMSE conjecture}\label{sec:thesetting}
 Let  $X$ and $N$ be  as stated in the introduction. The minimum mean squared
error along the heat flow is  defined as 
\[\MMSE(X, t) = \E[(X - \E[X \, | \, X_t])^2]\]
where $X_t$ is the heat flow starting from $X$ (recall \eqref{eq:heatfl}) and $t>0$.
 Straightforward computations yield 
 \[4t^2I(X,t) = 2nt - \MMSE(X,t)\] 
 for all $t>0$, hence $\MMSE$,  Fisher information  and thus entropy are all related along the heat flow.   Note also how  $\MMSE(X, t)$ is well defined irrespective of the properties of $X$.

In the sequel, it will be more convenient to work in the  additive
Gaussian channel standardization. More precisely, with  $Y_s := \sqrt
s X + N$ as in \eqref{eq:gaucha} for  $s>0$ we set  
\[\mmse(X, s) := \E[(X - \E[X \, | \, Y_s])^2]. \]  
The minimal mean squared error along the heat flow ($\MMSE$) and in
the Gaussian channel ($\mmse$) are equivalent, in the sense that
\[\forall s>0,\,\mmse(X,s) = \MMSE \left ( X,  \frac{1}{2s} \right ).\]
In line with the literature,  we will state all forthcoming results in terms of the function $\mmse$; our results  are, of course, also valid for the function $\MMSE$.

We can now reformulate the MMSE conjecture from the introduction
(given in its entropic form) in its original form as initially posed
in \cite{Ver} (whence the name of this conjecture).

\begin{mmseconj}[mmse form]\label{MMSEConj-mmse}
Let $X$ and $Y$ be two random variables. If $\mmse(X,\cdot)=\mmse(Y,\cdot)$ then for some $a\in \R$, $\P_{Y}=\P_{X+a}$ or $\P_{Y} = \P_{-X+a}$. 
\end{mmseconj}


In  \cite[Proposition 7]{Ver}, it is showed that if  $X$ possesses moments of all orders, then the function $s \mapsto \mmse(X,s)$ admits infinitely many right derivatives at $s=0$. The first  low order  derivatives   are then obtained, and they all  take the form of polynomial expressions involving moments of $X$. 
Using the $\Gamma$-calculus introduced in \cite{LedG},  a full abstract description of the higher order  derivatives of the $\mmse$ function was obtained in \cite[Theorem 3]{Led}, in terms of iterated gradients and their brackets. In particular, it is     showed that under the assumption of infinite moments one has 
\begin{eqnarray}\label{formule-derivee}
&&\mmse^{(n)}(X, 0) = - (K_{n+1}(X))^2 - {R}_{n+1}(K_2(X),\ldots,K_{n}(X)),
\end{eqnarray}
where the derivatives are taken along the heat flow then evaluated at $s=0$,  the sequence $(K_i(X))_{i \ge 1}$ denotes the  cumulants of $X$ (see Section \ref{S2} for a definition), and ${R}_{n+1}\in\mathbb{R}[T_2,\ldots,T_{n}]$ are polynomials in $n-1$ variables which were not made explicit (see Proposition 2 in \cite{Led}).  

Now,  if  all the cumulants of $X$ are positive then, if the function $\mmse$ of a random variable $Y$, admitting moments of all orders, is known and equal to that of $X$,  the successive derivatives at 0 given through \eqref{formule-derivee}   can be used to  determine  a system of equations relating the cumulants of $Y$ to those of $X$. If the system \eqref{formule-derivee} has a unique solution,  then the cumulants of $X$ and those of $Y$ must be equal, and thus their laws as well, if the distribution of $X$ is moment-determined.  This leads to  the following partial answer to the MMSE conjecture.

\begin{theorem}\label{Ledmain}\cite[Corollary 8]{Led}
Let $X$ be a random variable  with non-negative cumulants whose law is, moreover, determined by its moments. Then, the knowledge of the function $\mmse(X,\cdot)$ completely characterises the cumulants of $X$, and hence its law.
\end{theorem}

Theorem \ref{Ledmain}  ensures that  two  centred random variables determined by their moments  having non-negative cumulants and the same $\mmse$ must have the same distribution. This conclusion   slightly differs  from the question asked in  Conjecture~\ref{MMSEConj-mmse}.   While the assumption of being determined by moments is fulfilled for most classical distributions, the requirement of having non-negative cumulants seems more restrictive (although it is satisfied by an  important class of  distributions including  symmetric infinitely divisible distributions, see Proposition \ref{CumulID} above). 

\subsection{Our results regarding the MMSE conjecture}

The more complete description of the polynomials given in  Theorem \ref{Main4} leads to a refinement of identity  \eqref{formule-derivee} which in turn allows us to prove the following result. 
 

\begin{theorem}\label{T1}
Let $X$ be $\sMD$ (see Definition \ref{smd}). Suppose that $X$ verifies 
\begin{equation}\tag{$\star$}\label{Hypostar}
\forall m\ge 3,\quad a_{2m}(K_2(X),\ldots,K_{2m-1}(X)) \neq 0.
\end{equation}
  If  $Y$ is a  symmetric random variable admitting moments of all orders such that $\mmse(Y,\cdot) = \mmse(X,\cdot)$, then $\P_X = \P_Y$.
\end{theorem}
 
Condition $(\star)$ is, as we shall see, satisfied in many interesting cases. For instance we immediately obtain the next corollary. 
\begin{corollary}
Let $X$ be $\sMD$ such that its family of even order cumulants $(K_{2n})$ is algebraically independent. For all symmetric random variables $Y$ admitting moments of all orders, if 
$\mmse(Y,\cdot) = \mmse(X,\cdot)$, then $\P_X = \P_Y$.
\end{corollary}

\begin{example}\label{ExDFS}
Consider the symmetric distribution with finite support 
\[\mu :=   \sum_{j=1}^{n} \frac{p_j}{2}( \delta_{x_j} + \delta_{-x_j})\]
where $x_1, \ldots, x_n$ are real numbers and $p_1, \ldots, p_n$ are
positive real numbers such that ${\sum}_{j=1}^n p_j = 1$. The moments
of $\mu$ are polynomials in the $p_i$ and the $x_i$, hence so are the
cumulants by Proposition \ref{Cumul1} above.  Then consider the Borel
set
${\cup}_{m\in\mathbb{N}^*} {\cup}_{Q \in \mathbb{\mathbb{Q}}_m[T_1,...,T_{2n}]} \big \{ (p, x) \in [0,1]^n\times\mathbb{R}^n  \, | \, Q(p_1, \ldots, p_n, x_1, \ldots, x_n) = 0 \big\}.$
For $m\in \mathbb{N}^*$ and $Q\in \mathbb{\mathbb{Q}}_m[T_1,...,T_{2n}]\setminus \{0\}$, the set 
$\mathcal{V}(Q) := \big \{ (p, x) \in [0,1]^n\times\mathbb{R}^n \, | \, Q(p_1, \ldots, p_n, x_1, \ldots, x_n) = 0 \big \}$
has measure $0$ for the Lebesgue measure on $\mathbb{R}^{2n}$, because it is included in the zero set of a non-zero polynomial. Hence a countable union of such sets is still of measure $0$. Therefore for almost every $(p_1,\ldots, p_n,x_1, \ldots, x_n)$, $\mu$ defines a symmetric distribution verifying the hypothesis of the previous corollary.
\end{example}

Interestingly, assuming now that $X$ is symmetric with non-negative cumulants, we obtain that $\mmse(X,\cdot)$ characterises the distribution of $X$ in a sense which is slightly stronger than the one adopted in \cite{Led}. Indeed, while Theorem \ref{Ledmain} applied in the setting of Conjecture \ref{MMSEConj-mmse} requires the positivity of both sequences of cumulants, we only impose assumptions on  those of $X$, as follows.

\begin{theorem}\label{T2}
Let $X$ be $\sMD$. Suppose that $X$ has positive cumulants of even orders. For every centred random variable $Y$ admitting moments of all orders, if $\mmse(X,\cdot) = \mmse(Y,\cdot)$, then $\P_X = \P_Y$.
\end{theorem}

The latter holds for instance in the case of a symmetric Laplace distribution. As already mentioned, the main difference with the result from \cite{Led} is that in Theorem \ref{T2} we do not need to impose assumptions on the cumulants of $Y$ for the characterisation to hold. 

\begin{example}
The Laplace distribution $\mathcal{L}_b$ of parameter $(0,b)$, with $b>0$, is the symmetric distribution with probability density function
$x\mapsto  \exp ( -{|x|}/{b} )/{2b}$. Its cumulants of even orders are 
$K_{2n} = {(2n)!}b^{2n}/{n}$
for every positive integer $n$.
\end{example}

Lastly, we may also handle the case of cumulants of alternating sign. To this end we introduce the condition
\begin{equation}\tag{$\star\star$}\label{Hypostarstar}
\forall n\in \mathbb{N}^{\star}, \quad K_{2n}(X)K_{2n+2}(X) < 0
\end{equation}
(Note that we need the strict inequality in \eqref{Hypostarstar}). Then the following result holds. 
\begin{theorem}\label{T3}
Let $X$ be $\sMD$. Suppose that $X$ verifies \eqref{Hypostarstar}. For every centred random variable $Y$ admitting moments of all orders such that $K_4(Y)K_4(X) > 0$, if $\mmse(X,\cdot) = \mmse(Y,\cdot)$, then $\P_X = \P_Y$.
\end{theorem}

The condition $K_4(Y)K_4(X) > 0$ in Theorem \ref{T3} can be
reformulated in terms of excess kurtosis of $X$ and $Y$: indeed, the
excess kurtosis of a random variable $Z$ is the quantity
$K_4(Z)/K_2(Z)^2$. Then the condition given is equivalent to requiring
that the excess kurtosis of $X$ and $Y$ have the same sign.


\begin{example}
 Let $B_n$ denotes the $n$-th Bernoulli number. 
 The  cumulants of the uniform distribution on $[-1,1]$  are given by
$K_n = {B_n}/{n}$
for every positive integer $n$. Therefore the uniform distribution on $[-1, 1]$ verifies \eqref{Hypostarstar}.
\end{example}

\begin{example}
The  cumulants of the Rademacher distribution are given by
$K_{2n}:= {2^{2n}(2^{2n}-1)}B_{2n}/{2n}$
for  every positive integer $n$. Therefore the Rademacher distribution verifies \eqref{Hypostarstar}.
\end{example}

\begin{remark}
Note that condition \eqref{Hypostarstar} implies condition \eqref{Hypostar}.
\end{remark}

\begin{remark}
Not all symmetric distributions  have positive cumulants of even orders nor verify  condition \eqref{Hypostarstar}. For $\lambda \in \, ]0,1[$, consider the distribution
\[\mu_{\lambda} := (1-\lambda^2)\delta_0 + \frac{\lambda^2}{2} \left ( \delta_{\frac{1}{\lambda}} + \delta_{-\frac{1}{\lambda}} \right )\] 
Set $\lambda=\frac{1}{2}$, then the distribution $\mu_{\frac{1}{2}}$ is symmetric and bounded. The respective cumulants of $\mu_{\frac 1 2}$ given by
\[K_2=1,\quad K_4 = \frac{1}{24}, \quad  K_6 = -\frac{7}{360}, \quad K_8 = \frac{53}{20160} \] 
are not all positive and do not respect \eqref{Hypostarstar}.
We observe numerically that the sign of the cumulants heavily depends on $\lambda$, but we still don't know whether Theorem~\ref{T2} is applicable (though we know that it is for almost every $\lambda$ in $\R$ with Example \ref{ExDFS}).
\end{remark}

\subsection{Proofs of Theorems \ref{T1}, \ref{T2} and \ref{T3}}

Let $X$ and $Y$ be two centred random variables admitting moments of
every order, such that $\mmse(X,.) = \mmse(Y,.)$.  Using
\eqref{formule-derivee}, the equality of the $\mmse$ functions of $X$
and $Y$ leads to the equality of their derivatives at $0$ at every
order, hence to
\begin{equation}\label{Rtilde}
K_n(X)^2 + R_{n}(K_2(X),\ldots,K_{n-1}(X)) = K_n(Y)^2+R_{n}(K_2(Y),\ldots,K_{n-1}(Y))
\end{equation}
for every integer $n\geq2$. The strategy for the proofs is to use
induction to derive from \ref{Rtilde} the equality of the cumulants of
$X$ and $Y$ (and thus of their distributions, when $X$ or $Y$ is
moment-determined, or equivalently cumulant-determined, see Section
\ref{mdet}). The idea is to assume the equality of the cumulants up to
order $n-1$ and to use \eqref{Rtilde} for indices $n$ and $n+1$ in
order to get the equality of the cumulants of order $n$. Once one has
derived the equality of the cumulants of order $n$, the proof will be
completed by induction.

In order to understand the necessity of the assumptions required on $X$ or $Y$, let us first investigate relation \eqref{Rtilde} for  $n=2$ and $n=3$. 
Evaluated at $n=2$, we get 
$K_2(X)^2 = K_2(Y)^2$. 
Since the cumulant of order $2$ is non-negative, it follows that
$K_2(X) = K_2(Y)$. If $X$ is constant, since $K_2(X) = \mbox{Var}(X)$, then
$Y$ is constant; if they are both centred, they have the same
distribution $\delta_0$. Let us therefore suppose that $X$ is not
constant, that is $K_2(X) \neq 0$, and continue the investigation by looking  at \eqref{Rtilde} for $n=3$, which gives 
\[K_3(X)^2 -2K_2(X)^3 = K_3(Y)^2 -2K_2(Y)^3.\] 
Then   $K_3(X)^2=K_3(Y)^2$.

We cannot pursue the induction without  imposing  assumptions on $X$ or $Y$. Suppose that $X$ is symmetric, which means that its odd cumulants are zero (see Proposition \ref{CumulSym}). Then  $K_3(Y) = K_3(X) =0$. Equality \eqref{Rtilde} for $n=4, 5$ gives, after simplifications,
\begin{equation}
K_4(X)^2 = K_4(Y)^2 \mbox{ and } K_5(X)^2=K_5(Y)^2. 
\label{carre4}
\end{equation}
Again, since $X$ is symmetric, one deduces  $K_5(Y) = K_5(X) =0$.   Equality \eqref{Rtilde} for $n=6$ and $n=7$ yields,  after simplifications,
\begin{equation}
-30K_4(X)^3+K_6(X)^2 = -30K_4(Y)^3+K_6(Y)^2
\label{R6}
\end{equation}
and 
\begin{equation}
\begin{aligned} 
1260K_2(X)K_4(X)^3-140K_4(X)^2K_6(X) - 42K_2(X)K_6(X)^2 = \\ 
1260K_2(Y)K_4(Y)^3-140K_4(Y)^2K_6(Y) - 42K_2(Y)K_6(Y)^2 + K_7(Y)^2
\end{aligned}
\label{R7}
\end{equation}
Observing that $42\times 30 = 1260$,  multiplying \eqref{R6} by $42K_2(X)$ (which is not zero because $X$ is not constant) and  subtracting the result from \eqref{R7}, we get after some further simplifications 
\begin{equation}
140K_4(X)^2(K_6(Y)-K_6(X)) =  K_7(Y)^2
\label{R7bis}
\end{equation}
Pursuing this lign of reasoning  leads to our three  Theorems.
\begin{proof2}
Suppose further that $Y$ is symmetric and that $X$ verifies \eqref{Hypostar}. We proceed by induction, first proving  equality of the cumulants of $X$ and $Y$ up to  order $7$. Since both variables $X$ and $Y$ are symmetric, one only needs to show $K_4(X)=K_4(Y)$ and $K_6(X)=K_6(Y)$. For this matter, one does not need the condition \eqref{Hypostar} on $X$. Indeed, from \eqref{R7bis} above, we get 
$140K_4(X)^2(K_6(Y)-K_6(X))=0$ hence, either $K_4(X)=0$ or $K_6(X)=K_6(Y)$. Observe that if $K_6(X)=K_6(Y)$, then $K_4(X)=K_4(Y)$ from \eqref{R6}. Consequently, let us suppose $K_4(X)=0$. First, using \eqref{carre4}, we get $K_4(Y)=K_4(X)=0$. Then, from \eqref{R6}, we get
$K_6(X)^2=K_6(Y)^2.$
Identity \eqref{Rtilde} for $n=8$ and $n=9$ gives,  after simplifications,
\begin{equation*}\label{R8}
K_8(X)^2=K_8(Y)^2
\end{equation*}
and
\begin{equation*}\label{R9bis}
-560K_6(X)^3-72K_2(X)K_8(X)^2=-560K_6(Y)^3-72K_2(Y)K_8(Y)^2.
\end{equation*}
Combining these last two identities, one deduces $K_6(X)=K_6(Y)$. So, in any case, we get
\begin{equation*}
K_4(X) = K_4(Y) \mbox{ and } K_6(X) = K_6(Y)
\end{equation*}
which initialises the induction for $n=7$ (recall that $K_7(X)=K_7(Y)=0$). 

Let now $n$ be a positive integer greater than $7$, and suppose that $K_l(X) = K_l(Y)$, for every $ l \in \llbracket 2, n-1 \rrbracket$. Let us show then that $K_n(X)=K_n(Y)$. This is trivially satisfied if $n$ is odd. Let us therefore suppose that $n$ is even
. Identity  \eqref{Rtilde} at index $n$
gives
\begin{equation}
K_{n}(X)^2 + R_{n}(K_2(X),\ldots,K_{n-1}(X)) = K_{n}(Y)^2 + R_{n}(K_2(Y),\ldots,K_{n-1}(Y))
\label{Rn}
\end{equation}
from which it follows that $K_{n}(X)^2 = K_{n}(Y)^2$. To conclude, we use equality \eqref{Rtilde} at index $n+1$:  from Theorem \ref{Main4} one can write 
\[R_{n+1} = -n(n+1)T_2T_n^2+a_nT_n + b_n\] 
where $a_n,b_n$ are polynomials in $\mathbb{Z}[T_2,\dots,T_{n-1}]$. Then it follows, by induction hypothesis, that 
\[a\,K_n(X) = a\,K_n(Y)\]
 where $a := a_n(K_2(X),\dots,K_{n-1}(X))$. Since \eqref{Hypostar} is verified for $X$, and $n\geq 8$, the coefficient $a$
 is non zero, so one can deduce that $K_n(X) = K_n(Y)$. 
 This concludes the induction.
\end{proof2}


\begin{proof1}
Suppose that $X$ is symmetric and has positive cumulants of even orders. In the same fashion as before, we want to prove Theorem \ref{T2} using induction. From \eqref{R7bis}, using that $K_4(X)^2> 0$, it follows that 
$K_6(X) \leq K_6(Y).$ Since $K_6(X)$ is positive, so to is  $K_6(Y)$ and thus  $K_6(X)^2 \leq K_6(Y)^2$. Suppose $K_4(X)=-K_4(Y)$. Then, from \eqref{R6}, we know that $-60K_4(X)^3+K_6(X)^2 = K_6(Y)^2$
which is only possible if $K_4(X)$ is negative. Necessarily, $K_4(X) = K_4(Y)$ using \eqref{carre4}, and then $K_6(X)^2 = K_6(Y)^2$. Since $K_6(X)^2 \leq K_6(Y)^2$ and both quantities are positive, one deduces $K_6(X) = K_6(Y)$, and it follows $K_7(Y)=K_7(X)=0$. Hence the cumulants of $X$ and $Y$ are equal up to  order $7$. 

Let now $n$ be a positive integer greater than $7$. Suppose 
 \begin{equation}
 \forall l \in \llbracket 2, n-1 \rrbracket \quad K_l(X) = K_l(Y)
 \label{H3}
 \end{equation}
In the proof of Theorem \ref{T1}, we showed that \eqref{H3} implies $K_n(X)^2=K_n(Y)^2$. Using equality \eqref{Rn} and Theorem \ref{Main4} we get, after simplifications from the induction hypothesis,
\begin{equation}
K_{n+1}(X)^2 + \tilde{a}_nK_n(X) = K_{n+1}(Y)^2+\tilde{a}_nK_n(Y)
\label{aux}
\end{equation}
where $\tilde{a}_{n} :=a_{n}(K_2(X),\ldots,K_{2n-1}(X))$. If $n$ is odd, then since $K_n(X)^2=K_n(Y)^2$ and $X$ is symmetric, we get $0=K_n(X)=K_n(Y)$. Let's suppose that $n$ is even. Then $K_{n+1}(X)=0$ so the equality \eqref{aux} writes
\[\tilde{a}_nK_n(X) = K_{n+1}(Y)^2 + \tilde{a}_nK_n(Y)\]
Suppose $K_n(Y) = -K_n(X)$. Then 
\[2\tilde{a}_nK_n(X) = K_{n+1}(Y)^2\]
All that remains is to prove that $\tilde{a}_nK_n(X)$ is negative. Indeed, $K_n(X)$ is positive and 
\begin{equation*}
\begin{aligned}
\tilde{a}_n = & -(2m+1)  \begin{pmatrix} 2m \\ m \end{pmatrix} K_{m+1}(X)^2 \\ & \quad - (-1)^{m+1}(4m+2)\sum_{k=1}^{\lfloor m/2 \rfloor}   \begin{pmatrix} 2m \\ k-1 \end{pmatrix}  K_{2m+2-2k}(X)K_{2k}(X)
\end{aligned}
\end{equation*}
where $n=2m$, is negative because every cumulant of $X$ in the expression above, at the exception of $K_{m+1}(X)$, is positive. In the end, $\tilde{a}_nK_n(X)$ is negative and thus is not the square of a real number. Necessarily, $K_n(X) = K_n(Y)$. This concludes the induction.
\end{proof1}


\begin{proof3}
Suppose that $X$ is symmetric and verifies $\eqref{Hypostarstar}$. As before, we want to prove the result by induction, and for this purpose, we need to show that the cumulants of $X$ and $Y$ are the same up to order $7$. In order to do so, suppose that the kurtosis of $Y$ has the same sign has the kurtosis of $X$, i.e. $K_4(X)$ and $K_4(Y)$ have the same sign. From \eqref{carre4}, one then deduces $K_4(X) = K_4(Y)$.  
Condition \eqref{Hypostarstar} implies that $K_4(X)$ is negative, and that $K_6(X)$ is positive. As before, using that $K_4(X)^2 \neq 0$ and \eqref{R7bis}, one has 
\begin{equation}\label{signK6}
0<K_6(X) \leq K_6(Y)
\end{equation}
From \eqref{R6}, we get $K_6(X)^2=K_6(Y)^2$ and then, using \eqref{signK6}, one deduces
 \[K_6(X)=K_6(Y)\] 
Finally, from \eqref{R7bis}, we get  $K_7(Y)=K_7(X)=0$. 

Let now $n$ be a positive integer greater than $7$. Suppose 
 \begin{equation}
 \forall l \in \llbracket 2, n-1 \rrbracket, \quad K_l(X) = K_l(Y)
 \label{H4}
 \end{equation}
 As in the proof of Theorem \ref{T1}, \eqref{H4} implies $K_n(X)^2=K_n(Y)^2$. The case $n$ odd being the same as before, we suppose that $n$ is even. Then we suppose $K_n(Y) = -K_n(X)$ and we get
 \[2\tilde{a}_nK_n(X) = K_{n+1}(Y)^2\]
It remains to show that, under assumption \eqref{Hypostarstar}, $\tilde{a}_nK_n(X)$ is negative and we can conclude the induction the same way as before.
Note that since $X$ verifies \eqref{Hypostarstar}, we get 
\begin{equation}
K_{2m}(X) =  (-1)^{m+1}|K_{2m}(X)|
\label{K}
\end{equation}
for all positive integer $m$, using that $K_2(X) = \operatorname{Var}(X)$ is positive. It follows that for every positive integer $m$ greater or equal than $4$, the quantity $\tilde{a}_{2m}$ writes
\begin{equation}
\begin{aligned} 
\tilde{a}_{2m} = & -(2m+1)  \begin{pmatrix} 2m \\ m \end{pmatrix} K_{m+1}^2(X) \\ & \quad - (-1)^{m+1}(4m+2)\sum_{k=1}^{\lfloor m/2 \rfloor}   \begin{pmatrix} 2m \\ k-1 \end{pmatrix}  |K_{2m+2-2k}(X)K_{2k}(X)| \end{aligned}
\end{equation}
 and is positive if $m$ is even, negative if $m$ is odd. Now write $n=2m$ with $m$ a positive integer. Then, if $m$ is odd, $K_n(X)$ is positive by \eqref{K} and $\tilde{a}_{n}$ is negative, so $\tilde{a}_nK_n(X)$ is negative. If $m$ is even, then $K_n(X)$ is negative by \eqref{K} and $\tilde{a}_{n}$ is positive, so $\tilde{a}_nK_n(X)$ is negative. Hence $\tilde{a}_nK_n(X)$ is negative. Necessarily, $K_n(X) = K_n(Y)$. This concludes the induction.
\end{proof3}

\begin{remark}
	From these proofs, one can use a more general condition on the cumulants of $X$ , for which one would get the same result on the $\MMSE$ conjecture as stated in  Theorem \ref{T2}:
	\begin{equation}\tag{$\natural$}
		\forall m \geq 3, \quad K_{2m}(X)a_{2m}(K_2(X), \ldots,K_{2m-1}(X)) < 0
	\end{equation}
\end{remark}







\section{Inductive formula for $R_n$}\label{Combi}

In this section, we prove Theorem \ref{Main4}. We start with some notations: we denote by $T_\alpha$, where $\alpha$ is an element of $\N^d$ for some positive integer $d$, the monomial $T_{\alpha_1} \cdots T_{\alpha_d}$. For integers $a$ and $b$, we also denote by $\llbracket a,b \rrbracket$ the set of integers between $a$ and $b$. 

In order to go further into the study of $R_n$, we give a little more details on the context (see \cite{Led} and \cite{LedG} for more). Theorem \ref{entr-deri} in \cite{Led} is stated as 
\begin{theorem}[\cite{Led}, Theorem 1] \label{thom1led}
    Let $X$ be a random variable admitting moments of every order. Then 
     \[\partial_t^k H(X,t) = (-2)^{k-1} \int_\R f_t(x) \tilde{\Gamma}_{k}(v_t)(x) \, \dint x\]
    where $\tilde{\Gamma}_k$ is a multilinear form, called iterated gradient.
\end{theorem}
For every every smooth function
$u:\mathbb{R} \longrightarrow \mathbb{R}$, one defines the iterated
gradients $(\tilde{\Gamma}_n)_{n\geq 1}$ recursively by
$\tilde{\Gamma}_{n+1} = [ \Delta + \Gamma, \tilde{\Gamma}_n ]$ where
$\tilde{\Gamma}_1 = \Gamma$ with $\Gamma(f,g) = f'g'$, $\Delta$ is the
Laplacian operator and $[.,.]$ is a Lie bracket defined in
\cite[Section 3]{LedG} as follows: if $A$ is a $m$-multilinear
symmetric map and $B$ is a $k$-multilinear symmetric map on $\R$, then
for every $x \in \R^{m+k-1}$
\begin{equation*}
\begin{aligned}
2[A,B](x) = &   \frac{1}{(N)!} \sum_{\sigma \in \mathfrak{S}_{N}} \left ( mA(B(x_{\sigma(1)},\ldots, x_{\sigma(k)}), x_{\sigma(k+1)}, \ldots, x_{\sigma(N)}) \right.  \\ & \quad \quad \quad \quad \quad -  \left. kB(A(x_{\sigma(1)},\ldots, x_{\sigma(m)}), x_{\sigma(m+1)}, \ldots, x_{\sigma(N)}) \right )
\end{aligned}
\end{equation*}
where $N = m+k-1$, $x = (x_1,\ldots,x_N)$ and  $\mathfrak{S}_N$ is the group of permutations of $\llbracket 1,N \rrbracket$.

Then, in \cite{Led}, it is showed that
$$\tilde{\Gamma}_{n}(u) = \tilde{R}_{n}( u^{(2)},\dots,u^{(n)})$$ where
$\tilde{R}_n \in \mathbb{R}[T_2,\ldots,T_n]$ is a polynomial that
decomposes as $\tilde{R}_n = T_n^2 + R_n$ where $R_n$ belongs to
$\mathbb{R}[T_2,\ldots,T_{n-1}]$. This whole process will be helpful
for Lemma \ref{inductionR} below.

\begin{definition}
    Let $n$ be a positive integer greater than $1$. A partition of $2n$ is a tuple $\alpha = (\alpha_1,\ldots, \alpha_r)$ of integers between $2$ and $n-1$, included, such that $\alpha_1 \geq \ldots \geq \alpha_r$ and $\alpha_1 + \cdots + \alpha_r = 2n$. The integer $r$ is called the length of the partition $\alpha$ and is denoted by $\# \alpha$.

    We  denote by $I_{r,n}$ the set of partitions of $2n$ of length $r$ and $I_n$ the set of partitions of $2n$ of lengths between $3$ and $n$, included.
\end{definition}

The polynomials $(R_n)$ have a specific form, given in the following result from \cite{Led}.

\begin{proposition}[\cite{Led}]\label{FormRn}
For every integer $n$ greater than $1$, the polynomial $R_n \in \mathbb{R}[T_2,\ldots,T_{n-1}]$  from \eqref{formule-derivee} is of the form  
\begin{equation}\label{RN}
R_n =  \underset{\alpha \in I_{n}}{\sum} c_{\alpha}^{(n)} T_{\alpha}
\end{equation}
where $c_{\alpha}^{(n)} \in \mathbb{R}$. We use the convention $R_2 = 0$.
\end{proposition}

By making the first term of each partition explicit in \eqref{RN}, we  get
\begin{equation*}
R_{n+1} = \sum_{k=0}^{n-2} a_{n,k}T_{n-k}
\end{equation*}
where $a_{n,k}$ is a polynomial in $\R[T_2, \ldots,T_{n-k}]$. Since $(n,n,2)$ is the only element of $I_n$ with length $3$ and first twos terms equal to $n$, we can write
\begin{equation*}
R_{n+1} = c^{(n+1)}_{n,n,2}T_2T_n^2 + a_nT_n + \sum_{k=1}^{n-2} a_{n,k}T_{n-k}
\end{equation*}
where $a_n := a_{n,0}-c^{(n+1)}_{n,n,2}T_2T_n$ belongs to $\R [T_2, \ldots, T_{n-1}]$.

In order to prove Theorem \ref{Main4}, we need to further the description of $R_n$, in particular, be able to say something about its coefficients. To that end, we  consider the following lemma, obtained by reformulating the proof of Proposition 9 in \cite{Led}.
\begin{lemma}[\cite{Led}, {Proposition 9}]\label{inductionR}
Let $n$ be an integer greater than $2$. Then  
\begin{equation}
R_{n+1} = A_n + L(R_n) + H(R_{n})
\label{RecFormula}
\end{equation} 
where $A_n = -{{\sum}_{k=1}^{n-1}}  \binom{n}{k} T_{1+k}T_{1+n-k}T_{n}$ (starting with $R_{2} = 0$), and $L$, $H$ are linear operators on multivariate polynomials, defined below in Subsection \ref{def-HL}. 
\end{lemma}

Lemma \ref{inductionR} indicates relations between the coefficients of $R_{n+1}$ and the coefficients of $A_n$, $L(R_n)$ and $H(R_n)$. By analysing these relations, one can hope to deduce inductions formulas for the coefficients of $R_n$, by drawing a (affine) connection between the coefficients of $R_{n+1}$ and $R_n$. In order to do so, we  have to know how the operators $L$ and $H$ previously mentioned, act on $R_n$. The idea is: for a monomial $M = T_{\alpha}$, the coefficient of $M$ in $R_{n+1}$ is the sum of the coefficients of $M$ in respectively $A_n$, $L(R_n)$ and $H(R_n)$. Hence, we have to know more about the coefficients of $L(R_n)$ and $H(R_n)$, and since $H$ and $L$ are linear operators, it suffices to know how $L$ and $H$ act on monomials, which are a basis of the multivariate polynomials.
\begin{remark}
    Identity  \eqref{RecFormula} allows to write an efficient algorithm from which we can compute in reasonable time the polynomials $R_n$ for moderate values of $n$, see  \cite{AlgoG} and Section \ref{appen-pol} in the appendix. 
\end{remark}

\subsection{Two operators on polynomials}\label{def-HL}
As anticipated, in this subsection we  define  the two linear
operators $H$ and $L$ on the vector space of multivariate polynomials
$E$ which admits as basis the monomials $T_{\alpha_1}\cdots
T_{\alpha_j}$ where $j\in \mathbb{N}$ and $\alpha_1,\ldots, \alpha_j$
are positive integers. Thus, to define a linear operator $J$ on $E$,
we only need to give the image of the monomials $T_{\alpha_1}\cdots
T_{\alpha_j}$ by $J$. If $J$ is a linear operator and $m$ is a
positive integer, we denote by $J^m$ the composition $J \circ
\cdots \circ J$ where $J$ appears $m$ times. 
\begin{definition}
  Let $\alpha$ be an element of $\N^d$ with $d$ a positive integer,
  and let $P$ be a multivariate polynomial. We denote by
  $\varphi_\alpha (P)$ the coefficient of the monomial $T_\alpha$ in
  $P$. For a positive integer $m$, we also denote by
  $\varphi_{T_m}(P)$ the (unique) polynomial $Q$ such that
\[P = QT_m + R\]
where $R$ does not depend on $T_m$.
\end{definition}
Before giving the definitions of $L$ and $H$, we define two auxiliary
operators on polynomials. Although at this stage it may not be very
clear why these preliminary operators are needed, we will see that
they come in handy to compute the coefficients with an algorithm (see
\cite{AlgoG}).
\begin{definition}
We introduce the following linear operators on monomials: 
\begin{align*}
    D_1(T_{\alpha_1}\cdots T_{\alpha_r}) & =   \sum_{k=1}^r T_{\alpha_1}\cdots T_{\alpha_k+1}\cdots T_{\alpha_r}\\ 
    D_2(T_{\alpha_1}\cdots T_{\alpha_r}) & =   \sum_{k=1}^r T_{\alpha_1}\cdots T_{\alpha_k+2}\cdots T_{\alpha_r} 
\end{align*}
for all $r$-uplets $(\alpha_1,\ldots, \alpha_r)$, for all $r \in \mathbb{N}^*$.
\end{definition}
We can now define $L$ and $H$. 
\begin{definition}\label{DefLH}
For all $r$-uplets $(\alpha_1,\ldots, \alpha_r)$, for all $r \in \mathbb{N}^*$, we have 
\begin{equation*}
\begin{aligned}
    H(T_{\alpha_1}\cdots T_{\alpha_r}) & = T_1D_1(T_{\alpha_1}\cdots T_{\alpha_r}) -  \frac{1}{2} \sum_{k=1}^r T_{\alpha_1}\cdots D_1^{\alpha_k}(T_1^2)\cdots T_{\alpha_r} \\ 
    &  =   \ -\frac{1}{2}\sum_{k=1}^r \sum_{l=1}^{\alpha_k-1} \begin{pmatrix} \alpha_k \\ l \end{pmatrix}T_{1+l}T_{1+\alpha_k-l} \prod_{\underset{i \neq k}{i=1}}^{r} T_{\alpha_i}
\end{aligned}
\end{equation*}
 and 
\[L = \frac{1}{2} \left ( D_1^2-D_2 \right ).\]

\end{definition}
\begin{remark}\label{ExpL}
    Note how if $\alpha_1, \dots, \alpha_r$ are positive integers then  \[L(T_{\alpha_1} \cdots T_{\alpha_r}) =   \underset{1 \leq i < j \leq r}{\sum} T_{\alpha_1} \cdots T_{\alpha_i +1} \cdots T_{\alpha_j+1} \cdots T_{\alpha_r}.\]  
    \end{remark}

In order to ease some of the notations, we introduce the next
definition. 

\begin{definition}
  Let $\alpha$ and $\beta$ be tuples of integers. We denotes by
  $l_{\beta,\alpha}$ the quantity $\varphi_\alpha(L(T_\beta))$ and
  denotes by $h_{\beta,\alpha}$ the quantity
  $\varphi_\alpha(H(T_\beta))$.
\end{definition}
\begin{remark}\label{signLH}
    From Definition \ref{DefLH} and Remark \ref{ExpL}, it follows that $l_{\beta,\alpha}$ is non-negative and that $h_{\beta,\alpha}$ is non-positive, for all tuples $\alpha$ and $\beta$.
\end{remark}

\begin{proof}[Proof of Lemma \ref{inductionR}]
Let $u:\R \rightarrow \R$ be a smooth function.  In the proof of Proposition $9$ in \cite{Led}, we have 
\begin{equation*}
\begin{aligned}
\tilde{\Gamma}_{n+1}(u) & = [\Delta + \Gamma, \tilde{\Gamma}_n](u) = [\Delta + \Gamma, \Gamma_n + V_n](u) \\ 
& = [\Delta, \Gamma_n](u) + [\Gamma , \Gamma_n](u) + [\Delta ,V_n](u) + [\Gamma, V_n](u)
\end{aligned}
\end{equation*}
The first term $[\Delta, \Gamma_n](u)$ is just $\Gamma_{n+1}(u) = (u^{(n+1)})^2$ (see $(3.3)$ in \cite{Led}), hence 
\[R_{n+1}(u^{(2)}, \ldots, u^{(n)}) = [\Gamma , \Gamma_n](u) + [\Delta ,V_n](u) + [\Gamma, V_n](u)\]
The second term writes
\begin{align*}
 [\Gamma , \Gamma_n](u) & = u'\left ( u^{(n)^2} \right )' - \left ( (u')^2 \right ) ^{(n)}u^{(n)} \\
& = -  \sum_{k=1}^{n-1} \binom{n}{k} u^{(1+k)}u^{(1+n-k)}u^{(n)} = A_n(u^{(2)},\ldots, u^{(n)}),
 \end{align*}
Now, for the third and fourth terms, one writes the expressions of the quantities $[\Delta, B](u)$ and $[\Gamma, B](u)$, where $B(u) = u^{(\alpha_1)}\cdots u^{(\alpha_j)}$ with $\alpha_1, \ldots, \alpha_j$ positive integers between $1$ and $n+1$. we get 
\begin{equation*}
\begin{aligned}
 2[\Delta, B](u) & = \left ( u^{(\alpha_1)}\cdots u^{(\alpha_j)} \right )'' -   \sum_{i=1}^j u^{(\alpha_1)} \cdots u^{\alpha_i+2} \cdots u^{(\alpha_j)} \\
 & = 2L(T_{\alpha_1} \cdots T_{\alpha_j})(u^{(2)},\ldots, u^{(n)})
 \end{aligned}
 \end{equation*}
 Since $[.,.]$ is bilinear, $L$ is linear and $V_n$ is a linear combination of such multilinear forms $B$, it follows that  
 \[[\Delta, V_n](u) = L(R_n)(u^{(2)},\ldots, u^{(n)}).\]
  Similarly 
  \[[\Gamma, V_n](u) = H(R_n)(u^{(2)},\ldots, u^{(n)}),\]
and the conclusion follows. 
 \end{proof}
 The relation between the coefficient of the monomials in $R_{n+1}$
 with variable $T_n$ can now be described in terms of the monomials in
 $A_n$, $L(R_n)$ and $H(R_n)$.

\subsection{Algorithmic scheme of the computations of coefficients.}\label{AlgoScheme}
In this subsection, we explain our algorithmic approach to the
computation of  the coefficients of $R_{n+1}$ using \eqref{inductionR}. Let $n$ be a positive integer and $\alpha$ an element of $I_{n}$. One wants to compute the coefficient $c^{(n+1)}_\alpha$. There are three mains steps. 

\begin{enumerate}
\item[ \underline{Step 1}] Describe $c_\alpha^{(n+1)}$ in terms of
coefficients of $R_n$. Since $H$ and $L$ are linear operators, we get 
(see Lemma \ref{relcoefRn})
\begin{equation} \label{eqCoef}
    \begin{aligned}
        c^{(n+1)}_\alpha = \varphi_\alpha(A_n) + \sum_{\beta \in B_\alpha^{(n)}} \left ( l_{\beta,\alpha} + h_{\beta,\alpha} \right )c^{(n)}_\beta
    \end{aligned}
\end{equation}
where $B_\alpha^{(n)}$ is the set whose elements are in $I_{n-1}$ such
that $l_{\beta,\alpha} \neq 0$ or $h_{\beta,\alpha} \neq 0$. The
relation \eqref{eqCoef} reads like this : the coefficient
$c^{(n+1)}_{\alpha}$ is obtained as a contribution from $A_n$ plus a
contribution of both $H$ and $L$ that map monomials $T_\beta$ in $R_n$
into polynomials whose coefficients in front of the monomial
$T_\alpha$ is non-zero.
\item[\underline{Step 2}] 
 Determine the set $B_{\alpha}^{(n)}$ of
multi-indices $\beta$ for which the coefficient of the monomial
$T_\alpha$ in either $L(T_\beta)$ or $H(T_\beta)$ is non-zero, and
determine the coefficient $l_{\beta,\alpha} + h_{\beta,\alpha}$
associated.
\item[\underline{Step 3}] From equality \eqref{eqCoef} and the
  information obtained in Step 2, we get an induction-type relation
  between $c_\alpha^{(n+1)}$ and the coefficients $c_\beta^{(n)}$. We  can
  then determine $c^{(n+1)}_\alpha$ with induction techniques.
\end{enumerate}
\subsection{Operations on monomials.}\label{opmonom} In this
subsection, we investigate  how $L$ and $H$ act on monomials. The
main question we are interested in is: 

\medskip

\noindent \textbf{Question:} if $M$ is a monomial, then what are the
monomials - with non-zero coefficient - of the polynomials $L(M)$ and
$H(M)$ ?

\medskip

More precisely, for $\alpha$ in $I_n$, if $T_\alpha$ is a monomial
with non-zero coefficient in either $L(T_\beta)$ or $H(T_\beta)$, for
some tuple $\beta$, what can we say about $\beta$. The goal is to
determine the relations between $\alpha$ and those $\beta$ in order to
properly proceed to an induction proof to determine the coefficient
$c_\alpha^{(n+1)}$.

We start with a definition.

\begin{definition}
The degree $\deg(M)$ of a monomial $M$ is its total degree, meaning that if we can write $M$ as a product of the form 
$cT_{\alpha_1}^{m_1} \cdots T_{\alpha_r}^{m_r}$
where $c$ is a non-zero real number, and $\alpha_1, \ldots, \alpha_r$ and $m_1,\ldots,m_r$ are positive integers, then  $\deg(M) = m_1+\dots+m_r$. 
The  degree $\deg(P)$ of a polynomial $P$ is  the maximum of the degrees of monomials $M$ whose coefficients in $P$ are non-zero.
\end{definition}
The first observation is the following.

\begin{lemma}\label{HLdegre}
  Let $r$ be an integer greater than $2$. The operator $H$ sends
  monomials of degree $r$ onto sums of monomials of degree $r+1$ and
  the operator $L$ sends monomials of degree $r$ onto sums of
  monomials of degree $r$.
\end{lemma}
\begin{proof}
Let $M$ be a monomial of degree $r$. Let $\alpha_1, \ldots, \alpha_r$ be positive integers such that $M = T_{\alpha_1} \cdots T_{\alpha_r}$. Then 
\[L(M) =   \underset{1 \leq i < j \leq r}{\sum} T_{\alpha_1} \cdots T_{\alpha_i +1} \cdots T_{\alpha_j+1} \cdots T_{\alpha_r}\]
where each monomial $T_{\alpha_1} \cdots T_{\alpha_i +1} \cdots T_{\alpha_j+1} \cdots T_{\alpha_r}$ has degree $r$. Hence $L(M)$ is a sum of monomials of degree $r$. Similarly
\[H(M) =   -\frac{1}{2}\sum_{k=1}^r \sum_{l=1}^{\alpha_k-1} \begin{pmatrix} \alpha_k \\ l \end{pmatrix}T_{1+l}T_{1+\alpha_k-l} \prod_{{i=1}, {i \neq k}}^{r} T_{\alpha_i} \]
where each monomial $T_{1+l}T_{1+\alpha_k-l} \prod_{{i=1},{i \neq k}}^{r} T_{\alpha_i}$ has degree $r+1$. Hence $H(M)$ is a sum of monomials of degree $r+1$.
\end{proof}
\begin{remark}\label{HLlength}
     Let $\alpha$, $\beta$ be two tuples. Then, using Lemma \ref{HLdegre}, we get $l_{\beta,\alpha} = 0$ whenever $\# \beta \neq \# \alpha$, and $h_{\beta,\alpha} = 0$ whenever $\# \beta \neq \# \alpha -1$.
\end{remark}

By looking carefully at how the monomials in $H(T_\beta)$ and $L(T_\beta)$ (for a $\beta$ in $I_n$) are obtained, we get the two following results.

\begin{lemma}\label{indiceH}
Let $\alpha$, $\beta$ be elements of $I_{n}$ such that $h_{\beta,\alpha} \neq 0$. Then  $\beta_1 \geq \alpha_1$.
\end{lemma}
\begin{proof}
Denote by $r$ the quantity $\# \beta$. Consider the equality
\begin{equation}\label{auxH}
H(T_\beta) = \sum_{k=1}^{r} \sum_{l=1}^{\beta_k-1} \binom{\beta_k}{l} T_{1+l}T_{1+\beta_k-l} \prod_{\underset{j \neq k}{j=1}}^{r} T_{\beta_j}
\end{equation}
Since $h_{\beta,\alpha} = \varphi_\alpha(H(T_\beta)) \neq 0$, at least one of the monomial on the right-hand side of \eqref{auxH} is $T_\alpha$, by noticing that the coefficients of the polynomial above are all non-zero. By identifying the indices of the monomials, we get that $\alpha_1$ is equal to one of the $\beta_j$ for $j$ between $1$ and $r$ (which yields the result) or that there exists $k$ in $\llbracket 1, r \rrbracket$ and $l$ in $\llbracket 1, \beta_k-1 \rrbracket$ such that $\alpha_1$ is equal to either $1+l$ or $1+\beta_k-l$. In the latter case, we get  $\alpha_1 \leq \beta_k \leq \beta_1$.
\end{proof}

\begin{lemma}\label{indiceL}
Let $\alpha$, $\beta$ be elements of $I_{n}$ such that $l_{\beta,\alpha} \neq 0$. Then  $\beta_1 \leq \alpha_1$.
\end{lemma}
\begin{proof}
Denote by $r$ the quantity $\# \beta$. Consider the equality 
\begin{equation}\label{auxL}
L(T_\beta) = \sum_{1\leq i<j \leq r} T_{\beta_1} \cdots T_{\beta_i+1}\cdots T_{\beta_j+1} \cdots T_{\beta_r}
\end{equation}
Since $l_{\beta,\alpha} = \varphi_\alpha(L(T_\beta)) \neq 0$, at least one of the monomial $T_{\beta_1} \cdots T_{\beta_i+1}\cdots T_{\beta_j+1} \cdots T_{\beta_r}$ above on the right-hand side of \eqref{auxL} is equal to $T_\alpha$. That implies that there exists $i$ in $\llbracket 1,r \rrbracket$ such that $\alpha_i$ is equal to $\beta_1$ or $\beta_1+1$. Since $\alpha_1$ is the maximum of the $(\alpha_i)$, we get $\alpha_1\geq \beta_1$.
\end{proof}

The following lemma is just a property of calculus for $L$, used in the following section in the proof of Proposition \ref{degank}.

\begin{lemma}\label{ippL}
Let $P$ be a polynomial and $m$ be a positive integer. Then 
\[L(PT_m) = L(P)T_m+D(P)T_{m+1}\]
\end{lemma}
\begin{proof}
Since $L$ is a linear operator, one only needs to prove the equality for $P = T_\alpha$ for some $\alpha$ in $(\N^*)^d$ with $d$ a positive integer. Denoting $m$ by $\alpha_{d+1}$, one has 
\begin{displaymath}
\begin{aligned}
L(T_\alpha T_m) & = \sum_{1\leq i<j \leq d+1} T_{\alpha_1} \cdots T_{\alpha_i+1}\cdots T_{\alpha_j+1} \cdots T_{\alpha_{d+1}} \\ 
& = \sum_{1\leq i,j \leq d} T_{\alpha_1} \cdots T_{\alpha_i+1}\cdots T_{\alpha_j+1} \cdots T_{\alpha_d}T_{\alpha_{d+1}} \\ & \qquad \makebox[3cm]{\hfill} + \sum_{1\leq i \leq d} T_{\alpha_1} \cdots T_{\alpha_i+1}\cdots T_{\alpha_d}T_{\alpha_{d+1}+1}
\\ 
& = \left (\sum_{1\leq i,j \leq d} T_{\alpha_1} \cdots T_{\alpha_i+1}\cdots T_{\alpha_j+1} \cdots T_{\alpha_d} \right )T_{\alpha_{d+1}} \\ & \qquad \makebox[3cm]{\hfill} + \left ( \sum_{1\leq i \leq d} T_{\alpha_1} \cdots T_{\alpha_i+1}\cdots T_{\alpha_d} \right )T_{\alpha_{d+1}+1} \\
& = L(T_\alpha)T_m + D(T_\alpha)T_{m+1}
\end{aligned}
\end{displaymath}
\end{proof}

The next Proposition yields the induction formula between the coefficients of $R_{n+1}$ and $R_n$.

\begin{proposition}\label{relcoefRn}
Let $n$ be a positive integer and $\alpha$ an element of $I_{r,n+1}$, where $r$ is an integer such that $3 \leq r \leq n+1$. Then 
\begin{equation*}
\begin{aligned}
c_{\alpha}^{(n+1)} & = \sum_{\beta\in I_{r,n}} l_{\beta, \alpha} c_{\beta}^{(n)} + \sum_{\beta \in I_{r-1,n}} h_{\beta, \alpha} c^{(n)}_{\beta} + \varphi_{\alpha} (A_n)
\end{aligned}
\end{equation*}
\end{proposition}
\begin{proof}
Recall the equality 
\[R_{n+1} = L(R_n) + H(R_n) + A_n\]
from Lemma \ref{inductionR}. By applying $\varphi_\alpha$ on both sides, we get
\begin{equation*}
\begin{aligned}
c_\alpha^{(n+1)} & = \varphi_\alpha(R_{n+1}) = \varphi_\alpha(L(R_n)) + \varphi_\alpha(H(R_n)) + \varphi_\alpha(A_n) \\
& = \sum_{l=3}^n \sum_{\beta\in I_{l,n}} c_\beta^{(n)}\varphi_\alpha(L(T_\beta)) + \sum_{l=3}^n \sum_{\beta\in I_{l,n}} c_\beta^{(n)}\varphi_\alpha(H(T_\beta))+ \varphi_\alpha(A_n) \\
& = \sum_{l=3}^n \sum_{\beta\in I_{l,n}}l_{\beta,\alpha} c_\beta^{(n)}+ \sum_{l=3}^n \sum_{\beta\in I_{l,n}} h_{\beta,\alpha}c_\beta^{(n)}+ \varphi_\alpha(A_n) 
\end{aligned}
\end{equation*}
Since $l_{\beta,\alpha} = 0$ whenever $\# \beta \neq \# \alpha=r$, and since $h_{\beta,\alpha} = 0$ whenever $\# \beta \neq \# \alpha -1 = r-1$. With Remark \ref{HLlength}, we get the stated identity. 
\end{proof}

\begin{remark}\label{coefinteger}
    Since $R_2$ has integer coefficients, noticing that $l_{\beta,\alpha}$ and $h_{\beta,\alpha}$ are integers as well, induction shows that $R_{n+1}$ has also integer coefficients.  
\end{remark}

From Proposition \ref{relcoefRn}, we can determine the sign of the
coefficients of $R_n$.

\begin{corollary}\label{coefsign}
Let $n$ be an integer greater than $2$ and let $j$ be an integer such that $3\leq j \leq n$. Then for every $(\alpha_1, \ldots, \alpha_j)$ in $I_{j,n}$, the coefficient $c^{(n)}_{\alpha_1, \ldots, \alpha_j}$ is non-negative (resp. non-positive) if $j$ is even (resp. odd).
\end{corollary}
\begin{proof}
We prove this result by induction on $n$. The claim for  $n = 3$ follows because 
$R_3 = -2T_2^3$. 
 Let $n$ be an integer greater than $2$ and suppose that for every integer $j$ such that $3\leq j \leq n$ and for every $(\beta_1, \ldots, \beta_j)$ element of $I_{j,n}$, the coefficient $c^{(n)}_{\alpha_1, \ldots, \alpha_j}$ is non-negative (resp. non-positive) if $j$ is even (resp. odd). Let $j$ be an integer such that $3\leq j \leq n$ and let $\alpha = (\alpha_1, \ldots, \alpha_j)$ be an element of $I_{j,n+1}$. By Proposition \ref{relcoefRn}, we get
 \begin{equation*}
c^{(n+1)}_{\alpha} = \varphi_{\alpha}(A_n) +   \sum_{\beta \in I_{j,n}} l_{\beta,\alpha}c_{\beta}^{(n)}+\sum_{\beta \in I_{j-1,n}} h_{\beta,\alpha}c_{\beta}^{(n)}
\end{equation*}
Now, if $j \geq 4$, then $\varphi_\alpha(A_n) = 0$ since $A_n$ has degree $3$ and $\deg(T_\alpha) = \# \alpha = j > 3$. In the light of Remark \ref{signLH}, $l_{\beta,\alpha}$ is non-negative and $h_{\beta,\alpha}$ is non-positive. Using the induction hypothesis, we get the result wanted. \\
If $j=3$, then using the induction hypothesis like above and the fact that the coefficient of $A_n$ are non-positive, one concludes that $c_\alpha^{(n+1)}$ is also non-positive.
\end{proof}

\subsection{The degree of the polynomial ${a_{n,k}}$.}

Using the results from Subsection~\ref{opmonom}, we are able to
determine the degree of the polynomial $a_{n,k}$ defined right after
Proposition \ref{formRn}.

\begin{proposition}\label{degree2pk1}
Let $k$ be an integer between $0$ and $n-2$ and let $\alpha$ be an element of $I_{r,n+1}$ where $r>3+k$, such that $\alpha_1 = n-k$. Then $c_{\alpha}^{(n+1)} = 0$.
\end{proposition}
\begin{proof}
  We proceed by induction. Consider the predicate
\begin{align*}
P(n): & \quad \forall m \in \llbracket 5, n \rrbracket, \quad \forall k \in \llbracket 0, m-3 \llbracket, \quad \forall r \in \llbracket 4+k, m+1\rrbracket, \quad \forall \beta \in I_{r,m+1},\\ 
&  \quad \beta_1 = m-k \Rightarrow \quad c_\beta^{(m+1)}=0.
\end{align*}
It is verified for $n=5$. Indeed, 
\begin{equation*}
\begin{aligned}
R_6 = & \quad 120T_2^6 - 1200T_2^3T_3^2 + 210T_3^4 + 900T_2T_3^2T_4 + 300T_2^2T_4^2 -
30T_4^3 - 120T_3T_4T_5 \\ 
& - 30T_2T_5^2 \\ 
= & -30T_2T_5^2-120T_3T_4T_5 + (900T_2T_3^2 + 300T_2^2T_4 -
30T_4^2)T_4 \\
& +(-1200T_2^3T_3 + 210T_3^3)T_3 + 120T_2^6
\end{aligned}
\end{equation*}
Let $n$ be an integer greater than $5$, and suppose $P(n-1)$ is verified. Let $m$ be an integer between $0$ and $n/2$, $r$ an integer between $4+k\leq r\leq n+1$ and consider $\alpha$  in $I_{r,n+1}$ such that $\alpha_1 = n-k$. If $r = n+1$, then since $I_{n+1,n+1} =\{(2,\ldots,2)\}$, we get  $\alpha_1 = 2 = n-k$ i.e. $k = n-2$, which is not possible. Hence $r\leq n$. By Proposition \ref{relcoefRn}, it follows that 
\begin{equation}\label{aux12}
\begin{aligned}
c_{\alpha}^{(n+1)} & = \sum_{\beta\in I_{r,n}} l_{\beta, \alpha} c_{\beta}^{(n)} + \sum_{\beta \in I_{r-1,n}} h_{\beta, \alpha} c^{(n)}_{\beta} + \varphi_{\alpha} (A_n)
\end{aligned}
\end{equation}
where $l_{\beta, \alpha} = \varphi_\alpha ( L(T_\beta) )$ and $h_{\beta, \alpha} = \varphi_\alpha( H(T_\beta))$. One then proves that every term in the right-hand side of \eqref{aux12} is zero.

\underline{Step 1}: terms coming from $L$. \\
Let $\beta$ be an element of $I_{r,n}$ such that $l_{\beta,\alpha} \neq 0$. Then  $\# \beta = r \geq 4$. In particular,  $\beta_1 \leq n-k$ by Lemma \ref{indiceL}. By definition of $I_{r,n}$, one has $\beta_1 \leq n-1$ since $r>2$, hence $\beta_1 = n-1-k'$ for some $k' = k-1$ such that $0\leq k' \leq n-3-1 = (n-1)-3$. Furthermore, $r\leq n$. Using the induction hypothesis, we get $c_\beta^{(n)} = 0$.

\underline{Step 2}: terms coming from $H$.
 \\
 Let $\beta$ be an element of $I_{r-1,n}$ such that $h_{\beta,\alpha} \neq 0$. Then $\# \beta =  r-1 = r' \geq 3+k $. Also  $\beta_1 \geq n-k$ by Lemma \ref{indiceH}, and $\beta_1 \leq n-1$ by definition of $I_{r-1,n}$. Hence $\beta_1 = n-1-k'$ for some integer $k'=k-1$ such that $0\leq k' \leq (n-1)-3$ and $n \geq r' \geq 3+k = 4+k'$. By induction's hypothesis, one get $c^{(n)}_{\beta} = 0$.

\underline{Step 3}: terms coming from $A_n$. \\
The monomials in $A_n$ are of degree $3$. Since $\# \alpha = r > 3$, one has $\varphi_{\alpha}(A_n) = 0$.

It the follows that  $c_{\alpha}^{(n+1)} = 0$. Induction concludes the proof.

\end{proof}

Proposition \ref{degree2pk1} tells us that polynomial $a_{n,k}$ has
degree at most $2+k$. Indeed, the coefficient
$\varphi_\alpha(a_{n,k})$, where $\alpha_1 \leq n-k$, is the
coefficient $c^{(n+1)}_{\tilde{\alpha}}$ where
$\tilde{\alpha} = (n-k,\alpha_1, \ldots,\alpha_r)$ if $r = \#
\alpha$. Hence $\varphi_\alpha(a_{n,k})$ is zero as soon as
$\# \tilde{\alpha} = r+1 > 3 + k$, i.e. $r> 2+k$.

\begin{proposition}\label{degank}
Let $n$ be a positive integer, and $k$ an integer such that $0\leq k \leq n-2$. Then $\deg(a_{n,k}) = 2 + k$.
\end{proposition}
\begin{proof}
One proves this result by induction. Indeed, by Lemma \ref{inductionR}, we know 
\begin{equation*}
R_{n+1} = A_n + L(R_n)+ H(R_n)
\end{equation*}
Suppose $1\leq k < n-2$ and let $T_\alpha$ be the leading monomial in $a_{n-1,k}$. Then $c_\alpha^{(n)}$ is non-zero. We have 
\begin{equation*}
\begin{aligned}
L(R_n) & = L \left ( \sum_{l=0}^{n-3} a_{n-1,l}T_{n-1-l} \right ) =\sum_{l=0}^{n-3}  L \left (  a_{n-1,l}T_{n-1-l} \right ) \\
& = \sum_{l=0}^{n-3} L(a_{n-1,l})T_{n-1-l} + \sum_{l=0}^{n-3} D(a_{n-1,l})T_{n-l} \\
& = \sum_{l=1}^{n-2} L(a_{n-1,l-1})T_{n-l} + \sum_{l=0}^{n-3} D(a_{n-1,l})T_{n-l}
\end{aligned}
\end{equation*}
which leads to 
\begin{equation*}
\begin{aligned}
a_{n,k} & = L(a_{n-1,k-1}) + D(a_{n-1,k}) + \varphi_{T_{n-k}}(A_n + H(R_n)) \\ 
& = c_\alpha^{(n)}D(T_\alpha) + L(a_{n-1,k-1}) + D(a_{n-1,k}-c_\alpha^{(n)}T_\alpha) + \varphi_{T_{n-k}}(A_n + H(R_n))
\end{aligned}
\end{equation*}
Since by induction, $a_{n-1,k}$ has degree $2+k$, $T_\alpha$ has degree $2+k$, hence also $D(T_\alpha)$. Since the signs of the coefficients in $R_n$ only depend on their degree (see Corollary \ref{coefsign}), since $a_{n,k}$ has degree at most $2+k$ (see Lemma \ref{degree2pk1}), and since $c_\alpha^{(n)}$ is non-zero, one deduces that $a_{n,k}$ has degree $2+k$. \\

If $k = n-2$, then $a_{n,n-2} = (-1)^{n+1}n! \, T_2^{n}$ which has degree $n = k+2$. \\

If $k=0$, then using Lemma \ref{P1} below, which gives $c_{n,n,2}^{(n+1)}$, we get that $\varphi_{n,2}(a_{n,0})$ is not zero, hence $\deg(a_{n,0}) = 2 = 2+k$.

\end{proof}


\subsection{The form $a_{n}$.}
In this subsection, one determines explicitly the quadratic form
defined by $a_n := a_{n,0}-c^{(n+1)}_{n,n,2}T_2T_n$. In order to do
so, one only needs the coefficients of monomials of degree $3$ in
$R_{n+1}$ with $T_n$ as one of the variables.

First, we need to apply Proposition \ref{relcoefRn} in the context of
$a_n$.

\begin{lemma}\label{relcoefRnn}
    Let $\alpha$ be an element of $I_n$ such that $\# \alpha = 3$. Then 
    \begin{equation}
        c_{\alpha}^{(n+1)} = \varphi_{\alpha} (A_n) + \sum_{\beta\in I_{3,n}} l_{\beta, \alpha} c_{\beta}^{(n)} 
    \end{equation}
\end{lemma}

\begin{proof}
    It is just a matter of writing \eqref{relcoefRn} for $\alpha$ and noticing that $I_{2,n}$ is the empty set. 
\end{proof}

Set $\alpha = (n,b,c)$ and denote by $d^{(n+1)}_{n,b,c}$ the quantity $\varphi_{\alpha}(A_n)$. Then, by Lemma \ref{relcoefRnn}, 
\begin{equation}\label{coef4}
c_{n,b,c}^{(n+1)} = d^{(n+1)}_{n,b,c} + \sum_{\beta\in I_{3,n}} l_{\beta, (n,b,c)} c_{\beta}^{(n)} 
\end{equation}
Let's describe the elements $\beta$ of $I_{3,n}$ such that $l_{\beta,(n,b,c)} \neq 0$. Let $(\beta_1, \beta_2,\beta_3)$ be a triplet of positive integers ; one has 
\begin{equation}
L(T_{\beta_1}T_{\beta_2}T_{\beta_3}) = T_{\beta_1+1}T_{\beta_2+1}T_{\beta_3} + T_{\beta_1+1}T_{\beta_2}T_{\beta_3+1} + T_{\beta_1}T_{\beta_2+1}T_{\beta_3+1}.
\label{Lon3}
\end{equation}
From \eqref{Lon3}, one deduces that, if $\beta$ is an element of $I_{3,n}$,  the condition $l_{\beta,(n,b,c)} \neq 0$ is equivalent to that at least one of the equalities below is verified
\begin{enumerate}
\item[$\bullet$] $\{\beta_1+1,\beta_2+1, \beta_3\} = \{n,b,c\}$
\item[$\bullet$] $\{\beta_1+1,\beta_2, \beta_3+1\} = \{n,b,c\}$
\item[$\bullet$] $\{\beta_1,\beta_2+1, \beta_3+1\} = \{n,b,c\}$
\end{enumerate}
and, in this case, $l_{\beta,(n,b,c)}$ is equal to the number of those conditions verified. Since $\beta$ is an element of $I_{3,n}$, one has $\beta_1 \leq n-1$. Using the fact that $\beta_1 \geq \beta_2 \geq \beta_3$, and the equalities of sets above, one must have $\beta_1 = n-1$, if $l_{\beta,(n,b,c)} \neq 0$. \\

We will now compute the coefficients of monomials of degree $3$ in $R_{n+1}$ with variable $T_n$, using \eqref{coef4} to get induction formulas between coefficients of $R_{n+1}$ and of $R_n$.
\begin{lemma}\label{P1}
Let $n$ be an integer greater than $2$. Then $c_{n,n,2}^{(n+1)} = -n(n+1)$. Therefore 
\[ R_{n+1} = -n(n+1)T_2T_n^2 + a_nT_n + b_n.\]
\end{lemma}

\begin{proof}
Let $n$ be an integer greater than $2$. From  \eqref{coef4} we get 
\begin{equation}
\begin{aligned}
c_{n,n,2}^{(n+1)} & = d^{(n+1)}_{n,n,2} +   \sum_{\beta\in I_{3,n}} l_{\beta, (n,n,2)} c_{\beta}^{(n)}  \\
& = -2n +  \sum_{\beta\in I_{3,n}} l_{\beta, (n,n,2)} c_{\beta}^{(n)}.
\end{aligned}
\label{coefnn2}
\end{equation}
In order to describe more precisely the terms appearing in
\eqref{coefnn2}, we need to determine integers $\beta_2 ,\beta_3$ in
$\llbracket 2,n-1 \rrbracket$ such that $\beta_2 \geq \beta_3$,  
\begin{equation}
\beta_2 + \beta_3 = n+1
\label{partnn2}
\end{equation}
and such that at least one of these conditions
\begin{enumerate}
\item[$\bullet$] $\{n,\beta_2+1, \beta_3\} = \{n,2\}$
\item[$\bullet$] $\{n,\beta_2, \beta_3+1\} = \{n,2\}$
\item[$\bullet$] $\{n-1,\beta_2+1, \beta_3+1\} = \{n,2\}$
\end{enumerate}
is verified. Let $(n,\beta_2,\beta_3)$ be an element of $I_{3,n}$ such that $l_{\beta,(n,n,2)} \neq 0$. If $\beta_3 = n-1$ then $\beta_2=\beta_3=n-1$ and by \eqref{partnn2} we get $n=3$. In this case, 
\[R_4 = - 12T_2T_3^2 + 6T_2^4\]
where $c^{(4)}_{3,3,2} = -3\times4$. Let's suppose $n\geq 4$, therefore $\beta_3$ is not equal to $n-1$. Since  either $\beta_3 \in \{n,2\}$ or $\beta_3+1 \in \{n,2\}$, and $\beta_3 \neq n-1$, we deduce that $\beta_3 = 2$, because $\beta_3 \geq 2$. By \eqref{partnn2}, we get $\beta_2 = n-1$, and only 
\begin{enumerate}
    \item[$\bullet$] $\{n,\beta_2+1, \beta_3\} = \{n,2\}$
\end{enumerate}
is verified, i.e. $\varepsilon_{n-1,n-1,2}^{(n,n,2)} = 1$. Hence \eqref{coefnn2} yields 
\[c_{n,n,2}^{(n+1)} = -2n + c^{(n)}_{n-1,n-1,2}\]
for all $n$ integers greater than $3$. By induction, one deduces
\[c^{(n+1)}_{n,n,2} = -n(n+1)\] 
for all $n$ greater than $2$, by using the fact that $c^{(4)}_{3,3,2} = -12 = -3\times4$. This concludes the proof.
\end{proof}
The next $3$ lemmas are proved using the same ideas as for Lemma \ref{P1} but in a more involved way ; their proofs are also relegated to the appendix.
\begin{lemma}\label{P2}
Let $n$ be an integer greater than $4$. Then 
\[c^{(n+1)}_{n,n-1,3} = -(n-1)n(n+1)\] 
\end{lemma}

\begin{lemma}\label{P4}
Let $n$ be an integer greater than $7$. Then
\[c^{(n+1)}_{n,n+1-k,k+1} = -2(n+1) \binom{n}{k}\] 
for all $k \in \llbracket 2, \lfloor \frac{n}{2} \rfloor -1 \rrbracket$.
\end{lemma}

\begin{lemma}\label{an}
Let $m$ be an integer greater than $4$. Then 
\begin{enumerate}
\item $c^{(2m+2)}_{2m+1,m+2,m+1}  =  -2m \begin{pmatrix} 2m \\ m \end{pmatrix}$

\item $c^{(2m+1)}_{2m,m+1,m+1} = -(2m+1) \begin{pmatrix} 2m \\ m \end{pmatrix}$

\end{enumerate}
\end{lemma}

\begin{remark}
The condition \eqref{Hypostar} is related to the question of whether or not the vector $(K_2, \ldots, K_{2n})$ is a zero of $a_n$, where $(K_m)$ is a family of cumulants of a symmetric random variable $X$. A first observation one can make is that
\[a_{2n+1}(K_2,\dots,K_{2n}) = 0, \]
for all positive integers $n$. Indeed, if $n \geq 2$, the monomials in $a_{2n+1}$ are all of the form $T_iT_j$ with $i,j$ integers with an opposite parity. So when evaluating $a_{2n+1}$ on the cumulants of the distribution, we get a sum of terms $K_iK_j$ where at least $i$ or $j$ is odd. Since the distribution of $X$ is symmetric, which implies that all the cumulants of odd orders are zero, we get the result stated.
\end{remark}

Since in Section \ref{sn: a contd}, one makes asumptions on the zeros of $a_n$, it is important to have in mind the next result.

\begin{lemma}\label{annondeg}
For every integer $n$ greater than $4$, $a_n$ is a non-degenerate quadratic form on $\R^{n-3}$.
\end{lemma}

\begin{proof}
Notice that $a_n$ is in $\mathbb{Z}[T_3,\dots,T_{n-1}]$, so since it is a homogenous polynomial of degree $2$, it defines a quadratic form on $\mathbb{R}^{n-3}$. The monomials in $a_n$ are of form $T_iT_j$ where $i,j$ are positive integers between $3$ and $n-1$, such that $i+j = n+2$. It follows that the matrix of $a_n$ in the canonical basis of $\mathbb{R}^{n-3}$ is 
\[\begin{pmatrix} 0 & \hdots & 0 & u_{n,n-1} \\ \vdots & \reflectbox{$\ddots$} & \reflectbox{$\ddots$} & 0 \\ 0 & \reflectbox{$\ddots$} &  \reflectbox{$\ddots$} & \vdots \\  u_{n,3} & 0 & \hdots & 0\end{pmatrix}\] 
where $u_{n,k}$ denotes the coefficient of the monomial $T_{n+2-k}T_{k}$ in $a_n$. From the previous calculations, these coefficients are all strictly negative. One deduces that the matrix of $a_n$ is invertible, and so $a_n$ is non-degenerate.
\end{proof}

\begin{remark}
For $n\ge 4$ the quadratic form $a_n$ has isotropic vectors and thus is not definite. To see this, one can consider the subspace 
\[F:= \left \{ (x_1,\dots,x_{2m}) \in \mathbb{R}^{2m} \quad / \quad \forall j \in \llbracket m+1, 2m \rrbracket \quad x_j =0 \right  \}\] 
which is isotropic for $a_{2m}$, i.e. for every $x\in F$, one has $a_{2m}(x) = 0$.
\end{remark}

\begin{remark}
  Lemma \ref{annondeg} implies that $a_n$ has a lot of
  zeros. Nevertheless, it is not known whether or not a sequence of
  real numbers can be the cumulants of a random variable (see
  \cite{Nardo2009}). That leaves one wondering if it is possible that
  every sequence of cumulants of a (symmetric) random variable
  verifies \eqref{Hypostar}. This leads to an algebraic sufficient
  condition for which the MMSE conjecture holds, for symmetric random
  variables determined by their moments, as
    \begin{center} 
        ``Let $X$ be a symmetric random variable determined by its moments. Then $X$ verifies \eqref{Hypostar}.''
    \end{center}
\end{remark}

 \bibliographystyle{plain}

\appendix 

\section{Algorithm for the computation of the polynomials $R_n$}\label{appen-pol}

Using forthcoming Lemma \ref{inductionR}, we are able to compute the polynomials $R_n$ using the  algorithm that can be found in \cite{AlgoG}. This algorithm can only compute the polynomials explicitly  for moderate values of $n$, as the expressions become rapidly intractable. For instance, we have 
\begin{align*}
    R_9 = & -40320T_2^9 + 1693440T_2^6T_3^2 - 5927040T_2^3T_3^4 - 6350400T_2^4T_3^2T_4 \\ & - 846720T_2^5T_4^2 + 534240T_3^6 + 6894720T_2T_3^4T_4 \\ & + 9737280T_2^2T_3^2T_4^2 + 846720T_2^3T_4^3 + 2540160T_2^2T_3^3T_5 \\& + 3386880T_2^3T_3T_4T_5 + 211680T_2^4T_5^2 - 1746360T_3^2T_4^3 \\ &- 448560T_2T_4^4  - 1980720T_3^3T_4T_5 - 2903040T_2T_3T_4^2T_5 \\ & - 1028160T_2T_3^2T_5^2 - 423360T_2^2T_4T_5^2 - 105840T_3^4T_6\\ & - 846720T_2T_3^2T_4T_6  - 282240T_2^2T_4^2T_6 - 423360T_2^2T_3T_5T_6 \\ & - 28224T_2^3T_6^2 + 157500T_4^2T_5^2  + 60480T_3T_5^3 + 60480T_4^3T_6 \\ & + 259560T_3T_4T_5T_6 + 40320T_2T_5^2T_6  + 35532T_3^2T_6^2 \\ & + 30240T_2T_4T_6^2  + 40320T_3T_4^2T_7 + 30240T_3^2T_5T_7 \\ & + 40320T_2T_4T_5T_7 + 24192T_2T_3T_6T_7 + 2016T_2^2T_7^2 \\ & - 560T_6^3 - 2520T_5T_6T_7 - 756T_4T_7^2 - 630T_5^2T_8 \\ & - 1008T_4T_6T_8 - 504T_3T_7T_8  - 72T_2T_8^2.
\end{align*}
More explicit expressions can be found in the supplementary material (for $n$ up to 20). The number of terms in  the first few polynomials are reported in  Table \ref{tab:my_label}. Numerical explorations indicate that the number of terms in $R_n$ is given by $u_n-1$ where $u_n$ is  the number of degree sequences with degree sum $2n$ representable by a non-separable graph (see \cite{Rods}).  We have not attempted to prove this fact. 

\begin{table}[h]
    \centering
    \begin{tabular}{c|cccccccccc}
       $n$  & 3 & 4 & 5 & 6 & 7 & 8 & 9 & 10 & 15 & 20      \\
       \hline
        $\mathrm{length}(R_n)$    &  1 & 2 & 4 & 8 & 14 & 24 & 42 & 69 & 665 & 4555 \\
    \end{tabular}
    \vspace{3mm}
    \caption{Number of distinct terms in the polynomials $R_n$}
    \label{tab:my_label}
\end{table}
\begin{table}[h]
    \centering
    \begin{tabular}{c|cccccccc}
       $n$  & 10 & 12 & 15 & 17 & 20 & 21 & 22 & 23    \\
       \hline
        CPU (s)    &  0.07 & 0.46 & 6.34 & 35.5 & 422 & 932 & 2348 & 5121 \\
    \end{tabular}
    \vspace{3mm}
    \caption{Time in seconds to compute the polynomials $R_n$}
    \label{tab:my_label}
\end{table}

\begin{example}\label{R9ank}
    Here are the expressions of the $a_{8,k}$ from Remark \ref{ExR9}:
    \begin{equation*}
\begin{aligned}
    a_{8,1} = & - 2520T_5T_6 + 40320T_3T_4^2 + 30240T_3^2T_5 + 40320T_2T_4T_5 + 24192T_2T_3T_6 \\ 
    & - 756T_4T_7 + 2016T_2^2T_7 \\ \\
    a_{8,2} = & \quad 60480T_4^3 + 259560T_3T_4T_5 + 40320T_2T_5^2     - 105840T_3^4 - 846720T_2T_3^2T_4 \\
    & - 282240T_2^2T_4^2 - 423360T_2^2T_3T_5 + 30240T_2T_4T_6 + 35532T_3^2T_6- 28224T_2^3T_6 \\
    & - 560T_6^2 \\ \\
    a_{8,3} = & - 1980720T_3^3T_4 - 2903040T_2T_3T_4^2 + 2540160T_2^2T_3^3 + 3386880T_2^3T_3T_4 \\ 
    & + 211680T_2^4T_5- 1028160T_2T_3^2T_5+ 157500T_4^2T_5- 423360T_2^2T_4T_5 \\
    & + 60480T_3T_5^2 \\ \\
    a_{8,4} = & \quad 6894720T_2T_3^4- 6350400T_2^4T_3^2 - 846720T_2^5T_4+ 9737280T_2^2T_3^2T_4 \\
    & + 846720T_2^3T_4^2  - 1746360T_3^2T_4^2- 448560T_2T_4^3 \\ \\
    a_{8,5} = & \quad 1693440T_2^6T_3 - 5927040T_2^3T_3^3  + 534240T_3^5 \\ \\
    a_{8,2} = & -40320T_2^8.
\end{aligned}
\end{equation*}
\end{example}

\section{Proofs from Subsection \ref{S2}}


\begin{proof}[Proof of Proposition \ref{Cumul2}]
    Let $\lambda$ and $a$ be real numbers. Using the independence of $X$ and $Y$, one has
    \begin{equation*}
        \begin{aligned}
            K_n(X+\lambda Y) & = (-i)^n (\ln\varphi_{X+\lambda Y})^{(n)}(0) \\ \\
            & =  (-i)^n \big ( \ln(\varphi_{X}\varphi_{\lambda Y}) \big )^{(n)}(0) \\ \\
            & = (-i)^n (\ln\varphi_{X})^{(n)}(0)+(-i)^n(\ln\varphi_{\lambda Y})^{(n)}(0) \\ \\
            & = K_n(X)+\lambda^n(-i)^n(\ln\varphi_{Y})^{(n)}(0) \\ \\
            & = K_n(X)+\lambda^n K_n(Y).
        \end{aligned}
    \end{equation*}
    Setting $e_a(t) := e^{ita}$ and using Leibniz's formula, we get 
     \begin{equation*}
        \begin{aligned}
            K_n(X+a) & = (-i)^n (\ln\varphi_{X+a})^{(n)}(0) \\ \\
            & =  (-i)^n \big ( \ln(\varphi_{X}e_a) \big )^{(n)}(0) \\ \\
            & = (-i)^n (\ln\varphi_{X})^{(n)}(0)+(-i)^n(\ln e_a)^{(n)}(0) \\ \\
            & = (-i)^n (\ln\varphi_{X})^{(n)}(0) = K_n(X).
        \end{aligned}
    \end{equation*}
    for $n\geq 2$.
\end{proof}

\begin{proof}[Proof of Corollary \ref{CumulSym}]
Using the symmetry of $X$ and Proposition \ref{Cumul2}, we get
\[K_n(X) = K_n(-X) = (-1)^nK_n(X)\]
hence the result.
    
\end{proof}

\begin{proof}[Proof of Proposition \ref{CumulID}]
For this proof, we adapt the work done in \cite{Ari} in the case of the kurtosis to every cumulants. Let $N$ be a positive integer. If $X$ is constant almost surely then it needs no proof. Suppose $X$ non constant almost surely. Since $X$ is infinitely divisible, there exists $X_1, \ldots, X_N$ independent and identically distributed random variables such that 
\[X \overset{Law}{=} X_1 + \cdots + X_N.\]
Let $n$ be another positive integer. By Proposition \ref{Cumul2}, we have 
\begin{equation*}
K_n(X) = NK_n(X_1) = K_n(\sqrt[n]{N} X_1)
\end{equation*}
and more generally
\begin{equation}
K_m(\sqrt[n]{N}X_1) =  N^{\frac{m}{n}-1}K_m(X)
\label{K2}
\end{equation}
for every positive integer $m$, where $\frac{m}{n}-1 < 0$ whenever $m<n$.  Now we proceed by induction. The cumulant of order $2$ of $X$ is always non-negative. Let $n$ be a even positive integer greater than $2$ and suppose that for every even positive integer $k$ lower than $n$, the cumulant $K_k(X)$ is non-negative. Write $n$ as $2m$ with $m$ a positive integer greater than $1$. One has, by H\"older's inequality,
\[\E[|\sqrt[n]{N}X_1|^{m+1}] \leq \E[|\sqrt[n]{N}X_1|^2] \E[|\sqrt[n]{N}X_1|^{2m}] \leq \E[|\sqrt{N}X_1|^2]\E[|\sqrt[n]{N}X_1|^{n}]. \]
Using Proposition \ref{Cumul1}, Proposition \ref{Cumul2} and \eqref{K2}, we get
\[\E[|\sqrt[n]{N}X_1|^{n}] = K_n(\sqrt[n]{N}X_1) +  \sum_{k=1}^{\lfloor \frac{n-1}{2} \rfloor} \begin{pmatrix} n-1 \\ 2k-1 \end{pmatrix} N^{\frac{2k}{n}-1}K_{2k}(X)\E[(\sqrt[n]{N}X_1)^{n-2k}]\]
where, by induction's hypothesis, every term in the sum is non-negative. Denote $u_N$ the sum $\sum_{k=1}^{\lfloor \frac{n-1}{2} \rfloor} \begin{pmatrix} n-1 \\ 2k-1 \end{pmatrix} N^{\frac{2k}{n}-1}K_{2k}(X)\E[(\sqrt[n]{N}X_1)^{n-2k}]$. Then $u_N$ does not depend on $X_1$. Indeed, by Proposition \ref{Cumul1}, each factor $\E[(\sqrt[n]{N}X_1)^{n-2k}]$ can be described as a sum of cumulants of $\sqrt[n]{N}X_1$ with even orders lower than $n$, and from \eqref{K2}, one knows that every cumulant of $\sqrt[n]{N}X_1$ of order lower than $n$ is equal to the cumulant of $X$ with same order, times a negative power of $N$. From that we can deduce that $\underset{N\rightarrow +\infty}{\lim} u_N = 0$ and since 
\[\forall N \in \mathbb{N}^\star, \quad K_n(X) + u_N \geq \frac{\E[|\sqrt[n]{N}X_1|^{m+1}]}{\E[|\sqrt{N}X_1|^2]} \geq 0,\]
we get  $K_n(X) \geq 0$. 
\end{proof}

\begin{remark}
We can also prove Proposition \ref{CumulID} using the Levy-Khintchine formula for infinitely divisible distributions, and we actually get that for every positive integer $n$ greater than $2$
\[K_n(X) =   \int_{\R \setminus \{0\}} u^n \nu(du)\]
where $\nu$ is the Levy measure associated to $X$ by the Levy-Khintchine theorem (see for instance Theorem $7.4$ in \cite{Steu}).
\end{remark}

\begin{proof}[Proof of Proposition \ref{prop-cumul}]
We first tackle  uniform distribution. Consider the cumulant generating function $f$ of a random variable $X$ with distribution the uniform distribution on $[-1,1]$. Let $t$ be a positive real number ; then 
\[f(t) = \ln \E [e^{tX}] = \ln \left (\frac{\sinh(t)}{t} \right ) = \ln( \sinh(t)) - \ln(t)\]
A primitive function of $\ln \circ \sinh$ is the function $\frac{1}{\tanh}$. Since for $t < 2\pi$ 
\[\begin{aligned} 
\frac{1}{\tanh(t)} & = \frac{e^{2t}+1}{e^{2t}-1} = 1 + \frac{2}{e^{2t}-1} = 1 + \frac{1}{t}\frac{2t}{e^{2t}-1} 
\\ & =   1 + \frac{1}{t}   \sum_{n=0}^{+\infty} B_n\frac{(2t)^n}{n!} = 1 + \frac{1}{t} +   \sum_{n=1}^{+\infty} B_n\frac{2^nt^{n-1}}{n!}
\end{aligned}\]
where $B_n$ denotes the $n$-th Bernoulli number. Hence by integrating we get 
\[f(t) = t + \ln(t) +   \sum_{n=1}^{+\infty} \frac{B_n2^n}{n}\frac{t^n}{n!} - \ln(t) =   \sum_{n=2}^{+\infty} \frac{B_n2^n}{n}\frac{t^n}{n!}\]
for all $0<t<2\pi$, since $2B_1= -1$. Since $f$ is even and takes the value $0$ at $0$, we get 
\[f(t) =   \sum_{n=2}^{+\infty} \frac{B_n2^n}{n}\frac{t^n}{n!}\]
for every $t \in ]-2\pi, 2\pi[$, hence the result.

\medskip

We next tackle the Laplace distribution. Consider the cumulant generating function $f$ of a random variable $X$
with distribution the Laplace distribution of parameter $(0,b)$ ; for
$t$ a real number such that $|t|<\frac{1}{b^2}$, it holds that 
\[\begin{aligned} 
f(t) & = \ln \E[e^{tX}] = \ln \left ( \frac{1}{1-b^2t^2} \right ) = -\ln(1-b^2t^2) \\ & =   \sum_{n=1}^{+\infty} \frac{(b^2t^2)^n}{n} =   \sum_{n=1}^{+\infty} \frac{(2n)!b^{2n}}{n}\frac{t^{2n}}{(2n)!}
\end{aligned}
\]
Hence the result.

\medskip

Finally we tackle the Rademacher distribution. Let us consider the
cumulant generating function $f$ of the Rademacher distribution and
denote by $(K_{2n})$ its sequence of even order cumulants. It holds that 
\[\forall t \in \mathbb{R}, \quad f(t) = \ln(\cosh(t)) =   \sum_{n=1}^{+\infty} K_{2n}\frac{t^{2n}}{(2n)!}\]
Since the function $\tanh$ is the derivative of $f$ and since 
\[\forall t \in \mathbb{R}, \quad \tanh(t) =   \sum_{n=1}^{+\infty} a_{n}\frac{t^{2n-1}}{(2n-1)!}\]
where $a_n := \frac{2^{2n}(2^{2n}-1)}{2n}B_{2n}$, by integrating we deduce the result stated above.
\end{proof}

\section{Proofs from Subsection \ref{opmonom}} 

\begin{proof}[Proof of Lemma \ref{P2}]
Let $n$ be an integer greater than $4$. From \eqref{coef4} we get 
\begin{equation}
\begin{aligned}
c_{n,n-1,3}^{(n+1)} & = d^{(n+1)}_{n,n-1,3} +   \sum_{\beta \in I_{3,n}} l_{\beta, (n,n-1,3)} c^{(n)}_{\beta} \\
& = -n(n-1) +   \sum_{\beta \in I_{3,n}} l_{\beta, (n,n-1,3)} c^{(n)}_{\beta}.
\end{aligned}
\label{coefnn-13}
\end{equation}
In order to describe more precisely the terms appearing in \eqref{coefnn-13}, we  need to determine the elements $\beta = (\beta_1, \beta_2, \beta_3)$ of $I_{3,n}$ such that $l_{\beta,(n,n-1,3)} \neq 0$. That comes down to determine integers $\beta_2 ,\beta_3$ in $\llbracket 2,n-1 \rrbracket$ such that $\beta_2 \geq \beta_3$, 
\begin{equation}
\beta_2 + \beta_3 = n+1
\label{partnn-13}
\end{equation}
and such that at least one of these conditions
\begin{enumerate}
\item[$\bullet$]  $\{\beta_2+1, \beta_3\} = \{n-1,3\}$
\item[$\bullet$] $\{\beta_2, \beta_3+1\} = \{n-1,3\}$
\item[$\bullet$] $\{\beta_2+1, \beta_3+1\} = \{n,3\}$
\end{enumerate}
is verified. Let $(\beta_2, \beta_3)$ be such a pair of
integers. Since $n$ is greater than $4$, using the results in the
proof of Lemma \ref{P1} we get  that $\beta_3$ is different than
$n-1$, so $\beta_3 \leq n-2$. If $\beta_3 = n-2$ then by
$\eqref{partnn-13}, $ $\beta_2 = 3$, so since $\beta_2 \geq \beta_3$,
we get $n=5$. In that case, we have 
$c^{(6)}_{5,4,3} = -120 = -3\times 4 \times 5$. Let's now suppose $n$
is greater than $5$. Then $\beta_3 \leq n-3$, and by
\eqref{partnn-13}, we deduce 
$\beta_2 = n+1-\beta_3 \geq n+1-(n-3) = 4$. Therefore $\beta_2$ or
$\beta_2+1$ cannot be equal to $3$, hence $\beta_3 \in \{2,3\}$. If
$\beta_3 = 2$ then $\beta_2 = n-1$ by \eqref{partnn-13}, and only
\begin{enumerate}
    \item[$\bullet$] $\{\beta_2, \beta_3+1\} = \{n-1,3\}$
    \item[$\bullet$] $\{\beta_2+1, \beta_3+1\} = \{n,3\}$
\end{enumerate}
are verified, i.e. $l_{(n-1,n-1,2),(n,n-1,3)} = 2$. If $\beta_3=3$ then $\beta_2 = n-2$ by \eqref{partnn-13}, and only
\begin{enumerate}
\item[$\bullet$]  $\{n,\beta_2+1, \beta_3\} = \{n,n-1,3\}$
\end{enumerate}
is verified, i.e. $l_{(n-1,n-2,3),(n,n-1,3)} = 1$. Hence \eqref{coefnn-13} yields 
\[c^{(n+1)}_{n,n-1,3} = -n(n-1) + 2c^{(n)}_{n-1,n-1,2} + c^{(n)}_{n-1,n-2,3} = -3n(n-1) + c^{(n)}_{n-1,n-2,3}\] 
using Lemma \ref{P1}. Since $c^{(6)}_{5,4,3} = -120$, we get by induction
\[c^{(n+1)}_{n,n-1,3} = -(n-1)n(n+1)\]
for all $n$ integer greater than $4$. 
\end{proof}

\begin{proof}[Proof of Lemma \ref{P2}]
Let $n$ be an integer greater than $4$. From \eqref{coef4} we get 
\begin{equation}
\begin{aligned}
c_{n,n-1,3}^{(n+1)} & = d^{(n+1)}_{n,n-1,3} +   \sum_{\beta \in I_{3,n}} l_{\beta, (n,n-1,3)} c^{(n)}_{\beta} \\
& = -n(n-1) +   \sum_{\beta \in I_{3,n}} l_{\beta, (n,n-1,3)} c^{(n)}_{\beta}.
\end{aligned}
\label{coefnn-13}
\end{equation}
In order to describe more precisely the terms appearing in \eqref{coefnn-13}, we need to determine integers $\beta_2 ,\beta_3$ in $\llbracket 2,n-1 \rrbracket$ such that $\beta_2 \geq \beta_3$, 
\begin{equation}
\beta_2 + \beta_3 = n+1
\label{partnn-13}
\end{equation}
and such that at least one of these conditions
\begin{enumerate}
\item[$\bullet$]  $\{\beta_2+1, \beta_3\} = \{n-1,3\}$
\item[$\bullet$] $\{\beta_2, \beta_3+1\} = \{n-1,3\}$
\item[$\bullet$] $\{n-1,\beta_2+1, \beta_3+1\} = \{n,n-1,3\}$
\end{enumerate}
is verified. Let $(\beta_2, \beta_3)$ be such a pair of
integers. Since $n$ is greater than $4$, using the results in the
proof of Lemma \ref{P1} we get that $\beta_3$ is different than
$n-1$, so $\beta_3 \leq n-2$. If $\beta_3 = n-2$ then by
$\eqref{partnn-13}, $ $\beta_2 = 3$, so since $\beta_2 \geq \beta_3$,
we get $n=5$. In that case, we have
$c^{(6)}_{5,4,3} = -120 = -3\times 4 \times 5$.  Let's suppose $n$
greater than $5$. Then $\beta_3 \leq n-3$, and by
\eqref{partnn-13},one deduces
$\beta_2 = n+1-\beta_3 \geq n+1-(n-3) = 4$. Therefore $\beta_2$ or
$\beta_2+1$ cannot be equal to $3$, hence $\beta_3 \in \{2,3\}$. If
$\beta_3 = 2$ then $\beta_2 = n-1$ by \eqref{partnn-13}, and only
\begin{enumerate}
    \item[$\bullet$] $\{n,\beta_2, \beta_3+1\} = \{n,n-1,3\}$
    \item[$\bullet$] $\{n-1,\beta_2+1, \beta_3+1\} = \{n,n-1,3\}$
\end{enumerate}
are verified, i.e. $\varepsilon_{n-1,n-1,2}^{(n,n-1,3)} = 2$. If $\beta_3=3$ then $\beta_2 = n-2$ by \eqref{partnn-13}, and only
\begin{enumerate}
\item[$\bullet$]  $\{n,\beta_2+1, \beta_3\} = \{n,n-1,3\}$
\end{enumerate}
is verified, i.e. $\varepsilon_{n-1,n-2,3}^{(n,n-1,3)} = 1$. Hence \eqref{coefnn-13} yields 
\[c^{(n+1)}_{n,n-1,3} = -n(n-1) + 2c^{(n)}_{n-1,n-1,2} + c^{(n)}_{n-1,n-2,3} = -3n(n-1) + c^{(n)}_{n-1,n-2,3}\] 
using Lemma \ref{P1}. Since $c^{(6)}_{5,4,3} = -120$, we get by induction
\[c^{(n+1)}_{n,n-1,3} = -(n-1)n(n+1)\]
for all $n$ integer greater than $4$. 
\end{proof}

\begin{proof}[Proof of Lemma \ref{P4}]
Let $n$ be an integer greater than $7$ and let $k$ be an integer such that $2\leq k \leq \lfloor \frac{n}{2} \rfloor-1$. From \eqref{coef4} we get 
\begin{equation}
\begin{aligned}
c_{n,n+1-k,k+1}^{(n+1)} & = d^{(n+1)}_{n,n+1-k,k+1} +   \sum_{\beta\in I_{3,n}} l_{\beta,(n,n+1-k,k+1)} c^{(n)}_{\beta} \\
& = -2\binom{n}{k} +   \sum_{\beta\in I_{3,n}} l_{\beta,(n,n+1-k,k+1)} c^{(n)}_{\beta}
\end{aligned}
\label{coefnn+1-kk+1}
\end{equation}
Following in the same fashion as in the previous proofs, one wants to determine integers $ \beta_2 ,\beta_3$ in $\llbracket 2,n-1 \rrbracket$ such that $ \beta_2 \geq \beta_3$, 
\begin{equation}
\beta_2 + \beta_3 = n+1
\label{partnn+1-kk+1}
\end{equation}
and such that at least one of these conditions
\begin{enumerate}
\item[$\bullet$]  $\{\beta_2+1, \beta_3\} = \{n+1-k,k+1\}$
\item[$\bullet$]  $\{\beta_2, \beta_3+1\} = \{n+1-k,k+1\}$
\item[$\bullet$]  $\{n-1,\beta_2+1, \beta_3+1\} = \{n,n+1-k,k+1\}$
\end{enumerate}
is verified. Let $(\beta_2, \beta_3)$ be such a pair of integers. Since $k\geq 3$, we have $n+1-k \leq n-2$, and $k+1 \leq \lfloor \frac{n}{2} \rfloor < n-2$, where the last inequality is due to $n\geq5$. Indeed, we have 
\[\frac{n}{2} < n-2 \Leftrightarrow n < 2n-4 \Leftrightarrow n > 4\]
We deduce that neither $n$, $n+1-k$ nor $k+1$ is equal to $n-1$, hence the last condition above 
\[\{n-1,\beta_2+1, \beta_3+1\} = \{n,n+1-k,k+1\}\]
does not happen. Since $k \leq \lfloor \frac{n}{2} \rfloor -1 < \lfloor \frac{n}{2} \rfloor$, we have $n+1-k > k+1$, this combined with $\beta_2 \geq \beta_3$ leads to three alternatives :
\begin{enumerate}
    \item[$\bullet$] $\beta_2+1 = n+1-k$ and $\beta_3 = k+1$, i.e. $\beta_2 = n-k$ ;
    \item[$\bullet$] $\beta_2 = n+1-k$ and $\beta_3 + 1 = k+1$, i.e. $\beta_3 = k$ ;
    \item[$\bullet$] $\beta_2 = k+1$ and $\beta_3+1 = n+1-k$, i.e. $\beta_3 = n-k$.
\end{enumerate}
Notice that $n-k > k+1$ : indeed 
\begin{equation}
n-k \geq n- \left \lfloor \frac{n}{2} \right \rfloor +1 
\label{piche1}
\end{equation}
Then if $n=2m+1$, with $m$ a positive integer, \eqref{piche1} becomes 
\[n-k \geq m+2 > m \geq k+1\]
and if $n=2m$, with $m$ a positive integer, \eqref{piche1} becomes
\[n-k \geq m+2 > k+1\]
Hence $\beta_3 = n-k$ does not happen. Now we compute $l_{(n-1,n-k,k+1),(n,n+1-k,k+1)}$ and $l_{(n-1,n+1-k,k),(n,n+1-k,k+1)}$. 
We have
\[L(T_{n-1}T_{n-k}T_{k+1}) = T_{n}T_{n+1-k}T_{k+1} + T_{n}T_{n-k}T_{k+2} + T_{n-1}T_{n+1-k}T_{k+2}\]
and 
\[L(T_{n-1}T_{n+1-k}T_{k}) = T_{n}T_{n+2-k}T_{k} + T_{n}T_{n+1-k}T_{k+1} + T_{n-1}T_{n+2-k}T_{k+1}\]
Since $n+2-k = k+1$ is equivalent to $k=\frac{n+1}{2}$ and since $k+2 = n+1-k$ if and only if $k=\frac{n-1}{2}$, one deduces that, since $k\neq \frac{n-1}{2}$ 
\[l_{(n-1,n-k,k+1),(n,n+1-k,k+1)} = l_{(n-1,n+1-k,k),(n,n+1-k,k+1)} = 1\]
hence \eqref{coefnn+1-kk+1} becomes
\begin{equation}
 c_{n,n+1-k,k+1}^{(n+1)}= -2\binom{n}{k} + c_{n-1,n-k,k+1}^{(n)} + c_{n-1,n+1-k,k}^{(n)}
\label{rec1}
\end{equation}
Denote $c_{n,n+1-k,k+1}^{(n+1)}$ by $u_{n,k}$. Then \eqref{rec1} becomes
\begin{equation}
u_{n,k} = -2\binom{n}{k} + u_{n-1,k} + u_{n-1,k-1}
\label{rec2}
\end{equation}

Consider the double-sequence $(v_{n,k})$ defined by
\[v_{n,k} :=-2(n+1) \binom{n}{k} \]
for every integer $n$ greater than $7$ and every integer $k$ such that $2\leq k \leq \lfloor \frac{n}{2} \rfloor-1$. Then $(v_{n,k})$ verifies the same induction formula \eqref{rec2} as $(u_{n,k})$
\[v_{n,k} = -2\binom{n}{k} + v_{n-1,k} + v_{n-1,k-1} \]
Indeed, 
\begin{equation*}
\begin{aligned}
-2\binom{n}{k} -2n\binom{n-1}{k} -2n \binom{n-1}{k-1} & = -2\binom{n}{k} -2n \left ( \binom{n-1}{k} +\binom{n-1}{k-1} \right ) \\ \\
& = -2\binom{n}{k} -2n\binom{n}{k} = -2(n+1) \binom{n}{k}
\end{aligned}
\end{equation*}
using Pascal formula for binomial coefficients. What's more, using \eqref{R9}, we get
\begin{equation}
\begin{aligned}
u_{8,3} = -1008 = -2(8+1) \begin{pmatrix} 8 \\ 3 \end{pmatrix} = v_{8,3}
\end{aligned}
\label{u83}
\end{equation}
and one also has, by Lemma \ref{P2},
\begin{equation}
 u_{n,2} = -(n-1)n(n+1) = -2(n+1) \begin{pmatrix} n \\ 2 \end{pmatrix} = v_{n,2}
 \label{un2}
\end{equation}
for every integer $n$ greater than $7$. From \eqref{u83} and \eqref{un2}, one shows, by induction, that
$u_{n,k} = v_{n,k}$
for every integer $n$ greater than $7$ and every integer $k$ such that $2 \leq k \leq \lfloor \frac{n}{2} \rfloor -1 $.
\end{proof}

\begin{proof}[Proof of Lemma \ref{an}]
Let $m$ be an integer greater than $3$. As before, we decompose the proof in several  steps. 

\underline{Step 1}: Writing \eqref{coefnn+1-kk+1} for $n=2m$ and $k = m$, we get
\begin{equation}
\begin{aligned}
c_{2m,m+1,m+1}^{(2m+1)} & = -\begin{pmatrix} 2m \\ m \end{pmatrix} +   \sum_{\beta\in I_{3,2m}} l_{\beta,(2m,m+1,m+1)} c^{(2m)}_{\beta} \\
\end{aligned}
\label{coef2m}
\end{equation}
Following in the same fashion than in the previous proofs, one wants to determine the integers $ \beta_2 ,\beta_3$ in $\llbracket 2,2m-1 \rrbracket$ such that $ \beta_2 \geq \beta_3$, 
\begin{equation}
\beta_2 + \beta_3 = 2m+1
\label{part2m}
\end{equation}
and such that at least one of these conditions
\begin{enumerate}
\item[$\bullet$]  $\{\beta_2+1, \beta_3\} = \{m+1\}$
\item[$\bullet$]  $\{\beta_2, \beta_3+1\} = \{m+1\}$
\item[$\bullet$]  $\{2m-1,\beta_2+1, \beta_3+1\} = \{2m,m+1\}$
\end{enumerate}
Let $(\beta_2, \beta_3)$ be such a pair of integers. The last condition of the latter conditions is impossible: indeed, for it to be true, the following
\begin{equation*}
\beta_2 = 2m-1  \quad \text{and} \quad  \beta_3 = m \quad \text{, or} \quad \beta_3 = 2m-1 \quad \text{and} \quad \beta_2 = m
\label{pouq}
\end{equation*}
needs to be true as well. However, since 
\[2\leq\beta_3 \leq \beta_2 \leq 2m-1\]
 $\beta_3 = 2m-1 > m =\beta_2 $ cannot happen, and  using \eqref{part2m}, $\beta_2 = 2m-1$ implies $\beta_3 = 2 \neq m$. Then, the condition 
 $\{\beta_2+1, \beta_3\} = \{m+1\}$
 cannot happen as well  because it implies $\beta_2 = m < m+1 = \beta_3$. Hence $\beta_2 = m+1$, $\beta_3 = m$ and one deduces $l_{(2m-1,m+1,m),(2m,m+1,m+1)} = 1$. It follows from \eqref{coef2m} that
 \begin{equation}
c^{(2m+1)}_{2m,m+1,m+1} = -  \begin{pmatrix} 2m \\ m \end{pmatrix} + c^{(2m)}_{2m-1,m+1,m}
\label{coef2mbis}
 \end{equation}
holds. 

\underline{Step 2}: 
Writing \eqref{coefnn+1-kk+1} for $n = 2m +1$ and $k=m$, we get
\begin{equation}
\begin{aligned}
c_{2m+1,m+2,m+1}^{(2m+2)} = & -2\begin{pmatrix} 2m+1 \\ m \end{pmatrix}   +   \sum_{\beta \in I_{3,2m}} l_{\beta,(2m+1,m+2,m+1)} c^{(2m+1)}_{\beta}
\end{aligned}
\label{coef2m+1}
\end{equation}
Again, one wants to determine the integers $ \beta_2 ,\beta_3$ in $\llbracket 2,2m \rrbracket$ such that $ \beta_2 \geq \beta_3$, 
\begin{equation}
\beta_2 + \beta_3 = 2m+2
\label{part2m+1}
\end{equation}
and such that at least one of these conditions
\begin{enumerate}
\item[$\bullet$]  $\{\beta_2+1, \beta_3\} = \{m+2,m+1\}$
\item[$\bullet$]  $\{\beta_2, \beta_3+1\} = \{m+2,m+1\}$
\item[$\bullet$]  $\{2m,\beta_2+1, \beta_3+1\} = \{2m+1,m+2,m+1\}$
\end{enumerate}
Let $(\alpha_2, \alpha_3)$ be such a pair of integers. Since $m > 2$, we know that  $m+2 < 2m$. We  deduce that 
$\{2m,\alpha_2+1, \alpha_3+1\} = \{2m+1,m+2,m+1\}$
cannot happen. Suppose 
$\{\alpha_2+1, \alpha_3\} = \{m+2,m+1\}$
verified, then since $\alpha_2\geq \alpha_3$, we get $\alpha_2 = \alpha_3 = m+1$. Suppose 
$\{\alpha_2, \alpha_3+1\} = \{m+2,m+1\}$
verified, then for the same reason, we get $\alpha_2 = m+2$ and $\alpha_3 = m$. Now we  compute $l_{(2m,m+1,m+1),(2m+1,m+2,m+1)}$ and $l_{(2m,m+2,m),(2m+1,m+2,m+1)}$. Clearly 
\begin{equation*}
\begin{aligned}
L(T_{2m}T_{m+1}T_{m+1})  & =   2T_{2m+1}T_{m+2}T_{m+1} + T_{2m}T_{m+2}T_{m+2}
\end{aligned}
\end{equation*} 
and 
\begin{equation*}
\begin{aligned}
L(T_{2m}T_{m+2}T_{m})  & =  T_{2m+1}T_{m+3}T_{m} + T_{2m+1}T_{m+2}T_{m+1} + T_{2m}T_{m+3}T_{m+1}
\end{aligned}
\end{equation*}
hence $l_{(2m,m+1,m+1),(2m+1,m+2,m+1)}=2$ and $l_{(2m,m+2,m),(2m+1,m+2,m+1)}=1$. Then \eqref{coef2m+1} becomes
\begin{equation}
\begin{aligned}
c_{2m+1,m+2,m+1}^{(2m+2)} & = - 2 \begin{pmatrix} 2m+1 \\ m \end{pmatrix} + 2c^{(2m+1)}_{2m,m+1,m+1} + c^{(2m+1)}_{2m,m+2,m} \\ \\
& = -2(m+1) \begin{pmatrix} 2m+1 \\ m \end{pmatrix} + 2c^{(2m+1)}_{2m,m+1,m+1}
\end{aligned}
\label{coef2m+1bis}
\end{equation}
using Proposition \ref{P4} to compute the coefficient $c^{(2m+1)}_{2m,m+2,m}$, since $m-1 < \left \lfloor \frac{2m+1-1}{2} \right \rfloor$. 

\underline{Step 3}: we combine the relations of Step 1 and Step 2, in order to get an induction formula for $c_{2m+1,m+2,m+1}^{(2m+2)}$. Combining \eqref{coef2mbis} and \eqref{coef2m+1bis}, we get
\begin{equation}
\begin{aligned}
c_{2m+1,m+2,m+1}^{(2m+2)} & = - 2 \begin{pmatrix} 2m+1 \\ m \end{pmatrix} -2\begin{pmatrix} 2m \\ m \end{pmatrix} + 2c^{(2m)}_{2m-1,m+1,m} \\ \\  & = -4(m+1) \begin{pmatrix} 2m \\ m \end{pmatrix} +  2c^{(2m)}_{2m-1,m+1,m}
\end{aligned}
\label{coef2m+1rec}
\end{equation}
Let $u_{m+1}$ denote the coefficient $c^{(2m+2)}_{2m+1,m+2,m+1}$. Then \eqref{coef2m+1rec} becomes
\begin{equation}
u_{m+1} = -4(m+1) \begin{pmatrix} 2m \\ m \end{pmatrix} + 2u_m
\label{rec3}
\end{equation}
\underline{Step 4}: we determine $(u_m)_{m\geq4}$ using $\ref{rec3}$. Consider the sequence $(v_m)$ defined by
\begin{equation*}
v_m := -2m  \begin{pmatrix} 2m \\ m \end{pmatrix}
\end{equation*} 
for every integer $m$ greater than $3$. Then 
\begin{equation*}
\begin{aligned}
-4(m+1) \begin{pmatrix} 2m \\ m \end{pmatrix} + 2v_m & = -4(m+1) \begin{pmatrix} 2m \\ m \end{pmatrix} -4m \begin{pmatrix} 2m \\ m \end{pmatrix} = -4(2m+1)\begin{pmatrix} 2m \\ m \end{pmatrix} \\ 
& = -2(m+1) \times (2m+2)(2m+2)\frac{(2m)!}{((m+1)m!)^2} \\ 
& = -2(m+1)\begin{pmatrix} 2m+2 \\ m+1 \end{pmatrix} = v_{m+1}
\end{aligned}
\end{equation*} 
for every integer $m$ greater than $3$, hence $(v_m)$ also verifies \eqref{rec3}.  Then consider
\begin{equation}
\begin{aligned}
R_8 = & -56T_2T_7^2 + \left ( -560T_4T_5 - 336T_3T_6 \right )T_7 + b_8 
\end{aligned}
\label{RR7}
\end{equation}
where $b_8$ belongs in $\mathbb{Z}[T_2, \ldots, T_6]$. From \eqref{RR7} we get $u_4 = -560 = v_4$, and thus 
\begin{equation}
u_m = v_m =  -  2m \begin{pmatrix} 2m \\ m \end{pmatrix} 
\label{coef2m+1final}
\end{equation}
for every integer $m$ greater than $3$.

\underline{Step 5}: one determines $c^{(2m+1)}_{2m,m+1,m+1}$. Combining \eqref{coef2m+1final} and \eqref{coef2mbis}, we get
\begin{equation}
\begin{aligned}
c^{(2m+1)}_{2m,m+1,m+1} & = -  \begin{pmatrix} 2m \\ m \end{pmatrix} + c^{(2m)}_{2m-1,m+1,m} = -  \begin{pmatrix} 2m \\ m \end{pmatrix} + u_{m-1} \\ 
& = -(2m+1) \begin{pmatrix} 2m \\ m \end{pmatrix}
\end{aligned}
\label{coef2mfinal}
\end{equation}
for every integer $m$ greater than $4$. Finally from \eqref{R9} we get
\[c^{(9)}_{8,5,5} = -630 = -9\begin{pmatrix} 8 \\ 4 \end{pmatrix}\] 
hence \eqref{coef2mfinal} holds for $m=4$.

\end{proof}
\end{document}